\documentclass[a4paper]{article}
\usepackage{fullpage}
\usepackage{caption}

\usepackage{enumitem}
\usepackage{calc}

\usepackage{lmodern} 

\usepackage{graphicx} 

\usepackage{amsmath, accents}
\usepackage{amssymb,tikz}
\usepackage{pict2e}
\usepackage{amsthm}
\usepackage{amscd}
\usepackage{mathrsfs}
\usepackage{todonotes}
\usetikzlibrary{calc} 

\usepackage{yfonts}

\usepackage{marvosym}
\usepackage[percent]{overpic}

\newtheorem{theorem}{Theorem} [section]
\newtheorem{proposition}[theorem]{Proposition}	
	
\newtheorem{lemma}[theorem]{Lemma}

\newtheorem{remark}[theorem]{Remark}
\theoremstyle{definition}

\newtheorem*{notation}{Notation}

\DeclareMathOperator{\tr}{Tr}

\newcommand{\C}{\mathbb{C}}
\newcommand{\R}{\mathbb{R}}
\newcommand{\N}{\mathbb{N}}
\newcommand{\re}{\text{\upshape Re\,}}
\newcommand{\im}{\text{\upshape Im\,}}

\def\XXint#1#2#3{{\setbox0=\hbox{$#1{#2#3}{\int}$}
\vcenter{\hbox{$#2#3$}}\kern-.5\wd0}}

\usepackage{tikz}
\usetikzlibrary{arrows}
\usetikzlibrary{decorations.pathmorphing}
\usetikzlibrary{decorations.markings}
\usetikzlibrary{patterns}
\usetikzlibrary{automata}
\usetikzlibrary{positioning}
\usepackage{tikz-cd}
\tikzset{->-/.style={decoration={
				markings,
				mark=at position #1 with {\arrow{latex}}},postaction={decorate}}}
	
	\tikzset{-<-/.style={decoration={
				markings,
				mark=at position #1 with {\arrowreversed{latex}}},postaction={decorate}}}

\tikzset{
	master/.style={
		execute at end picture={
			\coordinate (lower right) at (current bounding box.south east);
			\coordinate (upper left) at (current bounding box.north west);
		}
	},
	slave/.style={
		execute at end picture={
			\pgfresetboundingbox
			\path (upper left) rectangle (lower right);
		}
	}
}

\usetikzlibrary{shapes.misc}\tikzset{cross/.style={cross out, draw, 
         minimum size=2*(#1-\pgflinewidth), 
         inner sep=0pt, outer sep=0pt}}

\usepackage{pgfplots}
	
\numberwithin{equation}{section}

\def\bigO{{\cal O}}

\usepackage[colorlinks=true]{hyperref}
\hypersetup{urlcolor=blue, citecolor=red, linkcolor=blue}

\begin{document}
\title{The multiplicative constant for the Meijer-$G$ kernel determinant}
\author{Christophe Charlier, Jonatan Lenells, Julian Mauersberger}

\date{{\small Department of Mathematics, KTH Royal Institute of Technology, \\ 100 44 Stockholm, Sweden.}}

\maketitle

\begin{abstract}
We compute the multiplicative constant in the large gap asymptotics of the Meijer-G point process. This point process generalizes the Bessel point process and appears at the hard edge of Cauchy--Laguerre multi-matrix models
and of certain product random matrix ensembles. 
\end{abstract}


\section{Introduction and main results} 
The Meijer-$G$ point process is a determinantal point process whose kernel is built out of Meijer $G$-functions. It appears in the study of the smallest squared singular values of certain product random matrices \cite{KZ2014, KKS2016} and in Cauchy multi-matrix models \cite{BGS2014, BB2015} in the limit of large matrix dimension (see below for more details). Of particular interest is the distribution of the smallest particle, or equivalently, the probability of finding no particle in the interval $[0,s]$, $s>0$. In some particular cases, this distribution is related to a system of partial differential equations \cite{TW1994,S2014}. 
In this work, we are interested in the tail behavior of this distribution as $s \to + \infty$, known as the large gap asymptotics. The study of such asymptotics was initiated in \cite{ClaeysGirSti}, where the first two terms in the asymptotic expansion were found, and then pursued in \cite{CLMMuttalib}, where the third term was obtained. 
The purpose of this paper is to obtain an explicit expression for the next term in the expansion, which is the multiplicative constant.

\paragraph{Meijer-$G$ point process.} Given $m,n,q,p \in \N= \{0,1,2,\ldots\}$ and $a_1,\ldots,a_p,b_1,\ldots,b_q \in \mathbb{C}$ such that
\begin{align*}
0 \leq m \leq q, \quad 0 \leq n \leq p \qquad \mbox{ and } \qquad a_{k}-b_{j} \neq 1, \, 2, \, 3, \, \ldots \quad \mbox{ for all } 1 \leq k \leq n, \mbox{ and } 1 \leq j \leq m,
\end{align*}
the Meijer $G$-function is defined by \cite[Eq. 16.17.1]{NIST}
\begin{align}\label{def of Meijer G}
G_{p,q}^{m,n}\left( \left. \begin{array}{c}
a_1,\ldots,a_p \\
b_1,\ldots,b_q
\end{array} \right| z \right)=\frac{1}{2\pi i} \int_{L} \frac{\prod_{j=1}^m \Gamma(b_j-t) \prod_{j=1}^{n} \Gamma(1-a_j+t)}{\prod_{j=m+1}^{q} \Gamma(1-b_j+t) \prod_{j=n+1}^{p}\Gamma(a_j-t)} z^t dt,
\end{align}
where $z \in \mathbb{C}$ and $z^{t} = |z|^{t} e^{i t \arg z}$ with $\arg z \in (-\pi,\pi)$. The integration contour $L$ separates the poles of $\Gamma(b_j-t)$, $j=1,\ldots,m$, from those of $\Gamma(1-a_j+t)$, $j=1,\ldots,n$. Furthermore, $L$ is unbounded, oriented upwards, and goes to $\infty$ in sectors of the complex plane such that the integral in \eqref{def of Meijer G} converges. For example, we can choose $L$ as follows:
\begin{itemize}
\item \vspace{-0.1cm} If $|\arg z| < \frac{\pi}{2}(2m+2n-q-p)$, then $L$ can be chosen to start at $-i\infty$ and end at $+i\infty$.
\item \vspace{-0.2cm} If $q>p$, then $L$ can be a loop starting and ending at $+\infty$.
\item \vspace{-0.2cm} If $q<p$, then $L$ can be a loop starting and ending at $-\infty$.
\end{itemize}
Let $r>q \geq 0$ be fixed integers and let 
\begin{align}\label{cond on nu and mu}
\nu_{1},\ldots,\nu_{r},\mu_{1},\ldots,\mu_{q}>-1.
\end{align}
The Meijer-$G$ point process is a determinantal point process on $\R^+ = (0, + \infty)$ whose kernel is given by
\begin{align}\label{Meijer G kernel}
\mathbb{K}(x,y) = \int_{0}^{1} G_{q,r+1}^{1,q}\left( \left. \begin{array}{c}
-\mu_{1},\ldots,-\mu_{q} \\
0,-\nu_{1},\ldots,-\nu_{r}
\end{array} \right| tx \right)G_{q,r+1}^{r,0}\left( \left. \begin{array}{c}
\mu_{1},\ldots,\mu_{q} \\
\nu_{1},\ldots,\nu_{r},0
\end{array} \right| ty \right) dt.
\end{align}
The kernel \eqref{Meijer G kernel} can equivalently be written as \cite[Eq. (1.14) and Proposition A.2]{ClaeysGirSti}
\begin{align}\label{def kernels}
\mathbb{K}(x,y) =  \int_{\gamma}\frac{du}{2\pi i} \int_{\tilde{\gamma}}\frac{dv}{2\pi i} \frac{F(u)}{F(v)} \frac{x^{-u}y^{v-1}}{v-u},
\end{align}
where $F$ is given in terms of the Gamma function $\Gamma$ by
\begin{align}\label{def of F}
F(z)=\frac{\Gamma(z)\prod_{j=1}^{q} \Gamma(1+\mu_j-z)}{\prod_{j=1}^r \Gamma(1+\nu_j-z)}.
\end{align}
The contours $\gamma$ and $\tilde{\gamma}$ in \eqref{def kernels} are disjoint, oriented upwards, and tend to infinity in sectors lying strictly in the left and right half-planes, respectively. Furthermore, they separate the poles of $\Gamma(z)$ from the poles of $\prod_{j=1}^{q} \Gamma(1+\mu_j-z)\prod_{j=1}^{r} \Gamma(1+\nu_j-z)$, see Figure \ref{fig: gamma and gammatilde}. 

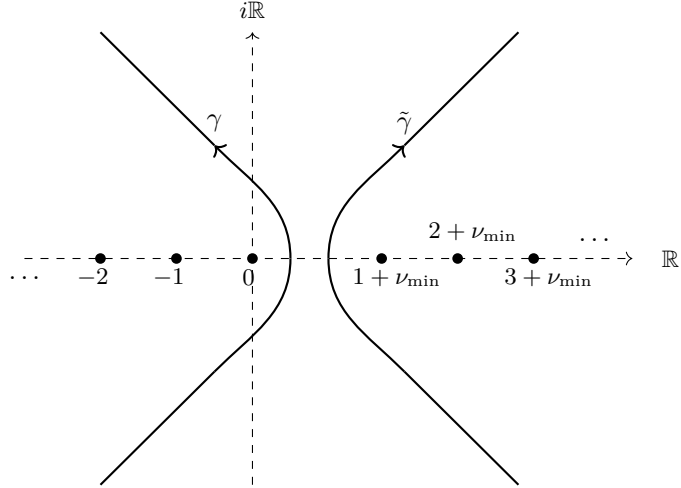
\begin{figure}
	\begin{center}
		\begin{tikzpicture}	
		
		\draw[dashed,->] (5,-3) -- (5,3);
		\draw[dashed,->] (2,0) -- (10,0);
		\node at (5,3.3) {\small $i\mathbb{R}$};
		\node at (10.5,0) {\small $\mathbb{R}$};
		
		\draw[thick] (4.5,-1.5) to [out=45, in=-90]  (5.5,0);
		\draw[thick,->] (5.5,0) to [out=90, in=-45]  (4.5,1.5);
		\draw[thick] (4.5,-1.5) to [out=-135, in=45] (3,-3);
		\draw[thick] (4.5,1.5) to [out=135, in=-45] (3,3);
		\node at (4.5,1.8) {$\gamma$};
		
		\draw[thick] (7,-1.5) to [out=135, in=-90]  (6,0);
		\draw[thick,->] (6,0) to [out=90, in=-135]  (7,1.5);
		\draw[thick] (7,-1.5) to [out=-45, in=135] (8.5,-3);
		\draw[thick] (7,1.5) to [out=45, in=-135] (8.5,3);
		\node at (7,1.8) {$\tilde \gamma$};
		
		\fill (5,0) circle (0.07cm);
		\node at (4.95,-0.25) {\small $0$};
		\fill (4,0) circle (0.07cm);
		\node at (3.9,-0.25) {\small $-1$};
		\fill (3,0) circle (0.07cm);
		\node at (2.9,-0.25) {\small $-2$};
		
		\fill (6.7,0) circle (0.07cm);
		\node at (6.9,-0.25) {\small $1+\nu_{\min}$};
		\fill (7.7,0) circle (0.07cm);
		\node at (7.9,0.35) {\small $2+\nu_{\min}$};
		\fill (8.7,0) circle (0.07cm);
		\node at (8.9,-0.25) {\small $3+\nu_{\min}$};
		\node at (2,-0.25) {$\ldots$};
		\node at (9.5,0.25) {$\ldots$};
		\end{tikzpicture}
		\caption{The contours $\gamma$ and $\tilde{\gamma}$, and $\nu_{\min}=\min \{ \nu_1,\ldots,\nu_r,\mu_{1},\ldots,\mu_{q} \}$.} 
		\label{fig: gamma and gammatilde}
	\end{center}
\end{figure}

\vspace{0.2cm}The Meijer-$G$ point process generalizes in a natural way the Bessel point process, which is the point process most commonly encountered at hard edges for random matrix ensembles. In fact, if $q=0$, $r=1$ and $\nu_{1}=\nu > -1$, then the kernel \eqref{Meijer G kernel} can be written as (see e.g. \cite[Section 5.3]{KZ2014})
\begin{align}\label{K reduces to Bessel here}
\mathbb{K}(x,y) = \left( \frac{y}{x} \right)^{\frac{\nu}{2}}\int_{0}^{1} J_{\nu}(2\sqrt{tx})J_{\nu}(2\sqrt{ty}) dt = 4 \left( \frac{y}{x} \right)^{\frac{\nu}{2}} \mathbb{K}_{\mathrm{Be}}(4x,4y),
\end{align}
where $\mathbb{K}_{\mathrm{Be}}$ is the kernel of the Bessel point process
\begin{equation}\label{Bessel kernel}
\mathbb{K}_{\mathrm{Be}}(x,y) = \frac{J_{\nu}(\sqrt{x})\sqrt{y}J_{\nu}^{\prime}(\sqrt{y})-\sqrt{x}J_{\nu}^{\prime}(\sqrt{x})J_{\nu}(\sqrt{y})}{2(x-y)},
\end{equation}
and $J_\nu$ is the Bessel function of the first kind of order $\nu$. 

\paragraph{Gap probability.} We consider the probability to observe a gap on $[0,s]$ in the Meijer-$G$ point process. It follows from the general theory of determinantal point processes \cite{Soshnikov} that this probability can be written as a Fredholm determinant:
\begin{align}\label{gap and Fredholm}
\mathbb{P}(\mbox{gap on }[0,s]) = \det\big(1-\mathbb{K}\big|_{[0,s]}\big).
\end{align}
It is well-known \cite{TW1994} that the distribution of the smallest particle in the Bessel point process (which corresponds to \eqref{gap and Fredholm} with $q = 0$ and $r=1$) is naturally expressed in terms of the solution of a Painlev\'{e} V equation. For $q=0$, $r \geq 2$ and integer values of $\nu_{1},\ldots,\nu_{r}$, the Fredholm determinant \eqref{gap and Fredholm} is instead related to a more involved system of partial differential equations \cite{S2014}.

\paragraph{Large gap asymptotics.} Asymptotics for \eqref{gap and Fredholm} as $s \to + \infty$ are referred to as \textit{large gap asymptotics}. For the Bessel point process, these asymptotics are known and given by: 
\begin{align}\label{Asymptotics Bessel}
\det\big( 1-\mathbb{K}|_{[0,s]}\big) = \frac{G(1+\nu)}{(2\pi)^{\frac{\nu}{2}}} \exp \bigg( -s+2\nu \sqrt{s} - \frac{\nu^2}{4} \ln(4s) + \bigO\big( s^{-1/2} \big)\bigg), \quad \mbox{as } s \to + \infty,
\end{align}
where $G$ denotes Barnes' $G$-function and $\mathbb{K}$ is as in \eqref{K reduces to Bessel here}.  The asymptotics \eqref{Asymptotics Bessel} were first conjectured by Tracy and Widom in \cite{TW1994}, then proved by Ehrhardt \cite{Ehr} for all complex values of $\nu$ in the strip $-1 < \re \nu < 1$, and finally proved by Deift, Krasovsky and Vasilevska \cite{DKV2011} for $-1< \re \nu$. Note that, because of the rescaling \eqref{K reduces to Bessel here}, the left-hand-side in \eqref{Asymptotics Bessel} is equal to $\det \big( 1-\mathbb{K}_{\mathrm{Be}}|_{[0,4s]}\big)$.

\medskip The study of the large gap asymptotics of the Meijer-$G$ point process (for general real values of the parameters) has been initiated by Claeys, Girotti and Stivigny in \cite{ClaeysGirSti}. They proved that there exist real constants $\rho$, $a$, $b$, $c$ and $C$ such that
\begin{align} \label{preliminary expansion fredholm det}
\det \big( 1- \mathbb{K}|_{[0,s]}\big)= C \exp\Big( -as^{2\rho}+bs^{\rho} + c \ln s+ \bigO \big( s^{-\rho}\big) \Big), \qquad \mbox{as } s \to + \infty,
\end{align}
and derived the following explicit expressions for $\rho$, $a$ and $b$:\footnote{Three point processes were considered in \cite{ClaeysGirSti}, and the corresponding quantities were denoted with an upperscript $(j)$, $j=1,2,3$. The coefficients \eqref{higher order constants11}-\eqref{higher order constants12} correspond to \cite[Theorem 1.2]{ClaeysGirSti} with $j=2$.}
\begin{align} 
& \rho=\frac{1}{1+r-q}, \qquad a=\frac{(r-q)^{\frac{1-r+q}{1+r-q}}(r-q+1)^2}{4}, \label{higher order constants11} \\
& b=(1+r-q)(r-q)^{-\frac{r-q}{1+r-q}} \bigg[ \sum_{j=1}^r \nu_j - \sum_{k=1}^{q} \mu_k \bigg]. \label{higher order constants12}
\end{align} 
These asymptotics have been extended to all orders in \cite{CLMMuttalib}, where it was shown that, for any $N \in \mathbb{N}$, there exist constants $C_{1},\ldots,C_{N} \in \mathbb{R}$ such that
\begin{align} \label{all order expansion}
\det \big( 1- \mathbb{K}|_{[0,s]}\big)= C \exp\bigg(  \hspace{-0.13cm}-as^{2\rho}+bs^{\rho} + c \ln s + \sum_{j=1}^{N} C_{j}s^{-j \rho} + \bigO \big( s^{-(N+1)\rho}\big) \bigg), \quad \mbox{as } s \to + \infty.
\end{align}
Furthermore, the following explicit expression for $c$ was obtained (see \cite[Remark 5.2]{CLMMuttalib}):
\begin{align} \label{higher order constants2}
c=\frac{r-q-1}{12(r-q+1)}-\frac{1}{2(r-q+1)} \bigg(  \sum_{j=1}^r \nu_j^2 - \sum_{k=1}^{q} \mu_k^2 \bigg).
\end{align}
Note that $C$ is a \textit{multiplicative} constant in \eqref{all order expansion}, which means that there is no precise description of the large gap asymptotics without its explicit expression. 

The problem of determining multiplicative constants in the asymptotics of Toeplitz, Hankel, and Fredholm determinants is a notoriously difficult problem in random matrix theory with a long history. The constant in the gap expansion of the sine kernel is known as the Dyson or Widom--Dyson constant (see e.g. \cite{DIKZ2007, EhrSine}); its occurrence in the asymptotic expansion was first established in 1973 by des Cloizeaux and Mehta \cite{dCM1973}, and its value was first determined (non-rigorously) by Dyson \cite{D1976} building on earlier work of Widom \cite{W1971}. Three different rigorous derivations of the Dyson constant were eventually presented in relatively rapid succession starting with the derivation by Krasovsky \cite{K2003} in 2004, followed by the derivations in \cite{EhrSine} and \cite{DIKZ2007}; these results also apply to certain related Toeplitz determinants.
Pioneering works where multiplicative constants are found for other determinants include \cite{DIK, BBD2008} for the Airy Fredholm determinant, \cite{Ehr, DKV2011} for the Bessel Fredholm determinant, and \cite{K2007, DIK FH} for Hankel and Toeplitz determinants with Fisher-Hartwig singularities. 



\medskip In this paper, we derive an explicit expression for the constant $C$ in \eqref{all order expansion} for general values of the parameters. The exact expression of $C$ combined with \eqref{higher order constants11}, \eqref{higher order constants12} and \eqref{higher order constants2} allows for an accurate description of the large gap asymptotics. The following theorem is our main result.

\newpage 

\begin{theorem}[Explicit expression for the constant $C$]\label{mainthm}
Let $r>q\ge0$ be fixed integers and let $\nu_1,\ldots, \nu_r,\mu_1,\ldots,\mu_q > -1$. The constant $C$ that appears in the asymptotic formula \eqref{preliminary expansion fredholm det} is given by
\begin{align}
C&=\frac{\prod_{j=1}^{r} G(1+\nu_j)\prod_{k=1}^{q} (2\pi)^{\frac{\mu_k}{2}}}{\prod_{k=1}^{q} G(1+\mu_k)\prod_{j=1}^{r} (2\pi)^{\frac{\nu_j}{2}}} \exp\big\{ (1-r+q)\zeta'(-1)  \big\}  \nonumber
\\
&\quad \times \exp \bigg\{  \bigg(\frac{1+r-q-(r-q)^2 }{2(1+r-q)} \bigg[\sum_{\ell=1}^{r}\nu_{\ell}^2-\sum_{k=1}^{q}\mu_k^2 \bigg] +\frac{-2+(r-q)^2(r-q-1)}{24(1+r-q)}  \nonumber
\\ 
& \qquad \quad   + \sum_{1\leq j<k\leq r}\nu_{j}\nu_{k}+ \sum_{1\leq j<k\leq q}\mu_{j}\mu_{k} - \sum_{j=1}^{r}\nu_{j} \sum_{k=1}^{q} \mu_{k}  +\sum_{k=1}^{q}\mu_k^2  \bigg)  \ln(r-q)\bigg\} \nonumber
\\
&\quad \times \exp\bigg\{ \bigg( -\frac{2-r+q}{2}\bigg[\sum_{\ell=1}^{r}\nu_{\ell}^2-\sum_{k=1}^{q}\mu_k^2\bigg] - \frac{(r-q-1)^2}{24} \nonumber \\
&\qquad \quad - \sum_{1\leq j<k\leq r}\nu_{j}\nu_{k}- \sum_{1\leq j<k\leq q}\mu_{j}\mu_{k} + \sum_{j=1}^{r}\nu_{j} \sum_{k=1}^{q} \mu_{k} -\sum_{k=1}^{q}\mu_k^2 \bigg) \ln(1+r-q) \bigg\},\label{constant C2}
\end{align}
where $G$ denotes Barnes' $G$-function and $\zeta'(-1)$ the derivative of Riemann's zeta function evaluated at $-1$.
\end{theorem}
 
The proof of Theorem \ref{mainthm} will be given in Section \ref{Section: proof main thm} after considerable preparations have been carried out in Sections \ref{Section: preliminaries}-\ref{Section: Asymptotics nu mu}. 

The value of $C$ was previously known for some very particular choices of the parameters. In Section \ref{section: consistency check}, we recall these expressions and show that Theorem \ref{mainthm} is consistent with them. At the end of Section \ref{section: consistency check}, we also provide two numerical checks of Theorem \ref{mainthm} for certain values of the parameters.

\subsection{Applications in random matrix theory}\label{section: applications}
The Meijer-$G$ point process appears at the hard edge scaling limit of several random matrix ensembles as the size of the matrices becomes large. 
In what follows, we explain how Theorem \ref{mainthm}, combined with the results of \cite{ClaeysGirSti,CLMMuttalib}, can be used to obtain information on the large gap asymptotics at the hard edge for these models. 

\paragraph{The Cauchy--Laguerre multi-matrix model.} 
The Cauchy two-matrix model has been introduced in \cite{BGS2009} and is a model for two positive definite Hermitian matrices coupled in a chain. The probability density function is defined over all pairs $(M_{1},M_{2})$ of two $n \times n$ positive definite Hermitian matrices, and it takes the form
\begin{align}\label{density function matrices M1 M2}
\frac{1}{Z_{n}} \frac{\det(M_{1})^{\alpha_{1}}\det(M_{2})^{\alpha_{2}}}{\det(M_{1}+M_{2})^{n}}e^{-\mathrm{Tr}(V_{1}(M_{1})+V_{2}(M_{2}))}dM_{1}dM_{2},
\end{align}
where $Z_{n}$ is the normalization constant, the two scalar potentials $V_{1}(x)$, $V_{2}(x)$ grow sufficiently fast as $x \to + \infty$, and the parameters $\alpha_{1}$ and $\alpha_{2}$ satisfy $\alpha_{1} >-1$, $\alpha_{2}>-1$, and $\alpha_{1} + \alpha_{2} > -1$. The eigenvalues $x_{1},\ldots,x_{n}$ of $M_{1}$ together with the eigenvalues $y_{1},\ldots,y_{n}$ of $M_{2}$ form a two-level determinantal point process. The Meijer-$G$ kernel \eqref{Meijer G kernel} was first discovered by Bertola, Gekhtman and Szmigielski in \cite{BGS2014}, where they considered the special case of $V_{1}(x) = V_{2}(x) = x$, known as the Cauchy-Laguerre two-matrix model. 
If we consider the point process involving only the eigenvalues of $M_{1}$, the associated limiting kernel as $n \to + \infty$ in the hard edge scaling limit, denoted by $\mathcal{G}_{01}$ in \cite{BGS2014}, is given by $( \frac{x}{y} )^{\alpha_{1}} \mathbb{K}(x,y)$, where $\mathbb{K}$ is given by \eqref{Meijer G kernel} with $r=2$, $q = 0$, $\nu_{1} = \alpha_{1} + \alpha_{2}$ and $\nu_{2} = \alpha_{1}$. Letting $x_{\min} := \min  \{x_{1},\ldots,x_{n}\}$, we have (see \cite[Appendix]{ClaeysGirSti})
\begin{align}\label{limiting distri C2M}
\lim_{n \to + \infty} \mathbb{P}(x_{\min}>\tfrac{s}{n}) = \det \big( 1- \mathbb{K}|_{[0,s]}\big),
\end{align}
and we can obtain the tail behavior as $s  \to + \infty$ up to and including the constant for the right-hand-side of \eqref{limiting distri C2M} by combining \eqref{all order expansion}, \eqref{higher order constants11}, \eqref{higher order constants12}, \eqref{higher order constants2} and \eqref{constant C2} with $r=2$, $q = 0$, $\nu_{1} = \alpha_{1} + \alpha_{2}$ and $\nu_{2} = \alpha_{1}$. 

\medskip

The Cauchy two-matrix model has been generalized to an arbitrary number $r$ of matrices in \cite{BB2015}, and similar results have been obtained at the hard edge as in the case of two matrices. Large gap asymptotics up to and including the constant can also be obtained using \eqref{higher order constants11}, \eqref{higher order constants12}, \eqref{higher order constants2} and Theorem \ref{mainthm}, with $q=0$ but general values of $r$, $\nu_{1},\ldots,\nu_{r}$.

\paragraph{Products of Ginibre matrices.} A complex Ginibre matrix is a random matrix whose entries are independent and identically distributed complex Gaussian variables. Let $G_1,\ldots,G_r$ be independent complex standard Ginibre matrices of size $(n+\nu_j) \times (n+\nu_{j-1})$, where $\nu_{0} = 0$ and $\nu_{1},\ldots,\nu_{r}$ are non-negative integers, and consider the product $G_r \cdots G_1$. If $r=1$, the squared singular values are well-studied and behave statistically as the eigenvalues of the Laguerre Unitary Ensemble, which are distributed according to a determinantal point process whose limiting kernel as $n \to + \infty$ in the hard edge scaling limit is given by the Bessel kernel \eqref{Bessel kernel} with $\nu = \nu_{1}$. For general $r \geq 2$, it is known from Akemann, Kieburg and Wei \cite{AKW2013} that the squared singular values of $G_r \cdots G_1$ still form a determinantal point process, and from Kuijlaars and Zhang \cite{KZ2014} that the limiting kernel in the hard edge scaling limit is the Meijer-$G$ kernel \eqref{Meijer G kernel} with $q=0$ (see also \cite{AI2015} for a detailed overview). If $x_{\min}$ denotes the smallest squared singular value of $G_r \cdots G_1$, then the limit \eqref{limiting distri C2M} holds, and we can obtain the tail behavior as $s  \to + \infty$ up to and including the constant for the right-hand-side of \eqref{limiting distri C2M} by combining \eqref{all order expansion}, \eqref{higher order constants11}, \eqref{higher order constants12}, \eqref{higher order constants2} and \eqref{constant C2}, and setting $q=0$.

\paragraph{Products of truncated unitary matrices.} Let $U_{1},\ldots,U_{r}$ be $r$ independent Haar distributed unitary matrices of size $\ell_{1} \times \ell_{1},\ldots ,\ell_{r} \times \ell_{r}$, respectively, and let $T_{j}$ be the upper left $(n + \nu_{j})\times (n + \nu_{j-1})$ truncation of $U_{j}$, $j=1,\ldots,r$. The parameters $\ell_{1},\ldots,\ell_{r}$ are positive integers, $\nu_{0} = 0$ and $\nu_{1},\ldots,\nu_{r}$ are non-negative integers. Furthermore, assume that $\ell_{1} \geq 2n + \nu_{1}$ and $\ell_{j} \geq n + \nu_{j}+1$ for $j \geq 2$. The squared singular values of the product $T_{r}\cdots T_{1}$ behave statistically as a determinantal point process \cite{KKS2016}. Taking $n \to + \infty$, we simultaneously have to let $\ell_{j} \to + \infty$ for $j=1,\ldots,r$. We choose a subset $J$ of indices
\begin{align*}
J = \{j_{1},\ldots,j_{q}\} \subset \{2,\ldots,r\}, \qquad \mbox{with } 0  \leq q = |J| < r
\end{align*}
and integers $\mu_{1},\ldots,\mu_{q}$ with $\mu_{k} \geq \nu_{k}+1$, and assume that
\begin{align*}
& \ell_{j}-n \to + \infty, & & \mbox{for } j \in \{1,\ldots,r\}\setminus J, \\
& \ell_{j_{k}}-n = \mu_{k}, & & \mbox{for } j_{k} \in J.
\end{align*}
Let $x_{\min}$ denote the smallest squared singular value of $T_{r}\cdots T_{1}$. It follows from \cite[Theorem 2.8]{KKS2016} and \cite[Appendix]{ClaeysGirSti} that
\begin{align}\label{limiting distri product truncated unitary}
\lim_{n \to + \infty} \mathbb{P}(x_{\min}>\tfrac{s}{c_{n}}) = \det \big( 1- \mathbb{K}|_{[0,s]}\big), \qquad \mbox{with } \quad c_{n} = n \prod_{j \notin J}(\ell_{j}-n),
\end{align}
and as in \eqref{limiting distri C2M} we can obtain asymptotics as $s \to + \infty$ for \eqref{limiting distri product truncated unitary} up to and including the constant by combining \eqref{all order expansion}, \eqref{higher order constants11}, \eqref{higher order constants12}, \eqref{higher order constants2} and \eqref{constant C2}. Note that in the above two models of product random matrices, we only need to utilize Theorem \ref{mainthm} for integer values of the parameters.

\subsection{Consistency checks of Theorem \ref{mainthm} and numerical confirmations}\label{section: consistency check}
We provide three different consistency checks of Theorem \ref{mainthm}; the first two verify consistency with known results in the literature for special choices of the parameters (for the constants $\rho$, $a$, $b$ and $c$, these checks were already carried out in \cite{ClaeysGirSti,CLMMuttalib}), while the third verifies consistency under the transformation $(r,q) \to (r+1, q+1)$ whenever $\nu_{r+1} = \mu_{q+1}$. For ease of explanation, we sometimes indicate the dependence of $\rho$, $a$, $b$, $c$, $C$ on the parameters $r$, $\nu_{1},\ldots,\nu_{r}$, $q$ and $\mu_{1},\ldots,\mu_{q}$ explicitly, e.g. for $C$ we write
\begin{align*}
C\big( (r; \nu_{1},\ldots,\nu_{r}), (q; \mu_{1},\ldots,\mu_{q}) \big).
\end{align*}
We also provide numerical confirmations of Theorem \ref{mainthm} at the end of this section.

\paragraph{Consistency with the large gap asymptotics of the Bessel point process.} Since our proof of Theorem \ref{mainthm} uses the asymptotic formula (\ref{Asymptotics Bessel}) for the Bessel kernel determinant, see Section \ref{outlinesubsec}, our work does not provide an independent derivation of the constant factor in (\ref{Asymptotics Bessel}); nevertheless, it is instructive to verify that Theorem \ref{mainthm} is consistent with (\ref{Asymptotics Bessel}). 
We verify from \eqref{higher order constants11}, \eqref{higher order constants12}, \eqref{higher order constants2} and \eqref{constant C2} that
\begin{align*}
& \rho\big( (1;\nu),(0;-)\big) = \frac{1}{2}, & &  a\big( (1;\nu),(0;-)\big) = 1, & & b\big( (1;\nu),(0;-)\big) = 2\nu, & & c\big( (1;\nu),(0;-)\big) = -  \frac{\nu^{2}}{4},
\end{align*}
and
\begin{align*}
C\big( (1;\nu),(0;-)\big) = \frac{G(1+\nu)}{(2\pi)^{\frac{\nu}{2}}} 2^{- \frac{\nu^{2}}{2}},
\end{align*}
which is indeed consistent with \eqref{Asymptotics Bessel}. 

\paragraph{Consistency with known results for the Muttalib--Borodin ensembles.} The Muttalib--Borodin ensembles \cite{M1995} with a Laguerre weight are joint probability density functions of the form
\begin{equation}\label{density function Muttalib Borodin}
\frac{1}{Z_{n}} \prod_{1 \leq j < k \leq n} (x_{k}-x_{j})(x_{k}^{\theta}-x_{j}^{\theta}) \prod_{j=1}^{n} x_{j}^{\alpha}e^{-x_{j}}dx_{j},
\end{equation}
where the $n$ points $x_{1},\ldots,x_{n}$ belong to the interval $[0,+\infty)$, $Z_{n}$ is a normalization constant and $\theta > 0$ and $\alpha>-1$ are two parameters of the model. These points behave statistically as a determinantal point process whose hard edge limiting kernel $\mathbb{K}^{\mathrm{MB}}$ can be written in terms of Wright's generalized Bessel functions \cite{B1999}. The corresponding large gap asymptotics are of the form \cite{ClaeysGirSti}
\begin{align*}
\det \big( 1- \mathbb{K}^{\mathrm{MB}}|_{[0,s]}\big)= C^{\mathrm{MB}} \exp\Big( -a^{\mathrm{MB}}s^{2\rho^{\mathrm{MB}}}+b^{\mathrm{MB}}s^{\rho^{\mathrm{MB}}} + c^{\mathrm{MB}} \ln s+ \bigO \big( s^{-\rho^{\mathrm{MB}}}\big) \Big), \qquad \mbox{as } s \to + \infty.
\end{align*}
The constants $\rho^{\mathrm{MB}}$, $a^{\mathrm{MB}}$ and $b^{\mathrm{MB}}$ have been obtained in \cite{ClaeysGirSti}, and $c^{\mathrm{MB}}$ and $C^{\mathrm{MB}}$ in \cite{CLMMuttalib}. For $q=0$ and certain particular choices of the parameters $r$, $\nu_{1}, \ldots,\nu_{r}$, $\alpha$ and $\theta$, the kernels $\mathbb{K}$ and $\mathbb{K}^{\mathrm{MB}}$ define the same point process (up to rescaling), see \cite[Theorem 5.1]{KS2014}. More precisely, if $r \geq 1$ is an integer, $\alpha > -1$ and
\begin{align}
& \theta = \frac{1}{r}, & & \nu_{j} = \alpha + \frac{j-1}{r}, \qquad j = 1,\ldots,r, \label{cond parameters for relation first and third kernel}
\end{align}
then the kernels $\mathbb{K}$ and $\mathbb{K}^{\mathrm{MB}}$ are related by
\begin{align*}
\left( \frac{x}{y} \right)^{\alpha}\mathbb{K}(x,y) = r^{r} \mathbb{K}^{\mathrm{MB}}(r^{r}x,r^{r}y).
\end{align*}
Therefore, if the parameters satisfy \eqref{cond parameters for relation first and third kernel}, we obtain the following relations:
\begin{align}
& \rho = \rho^{\mathrm{MB}}, \qquad a = a^{\mathrm{MB}} r^{2r\rho}, \qquad b = b^{\mathrm{MB}} r^{r\rho}, \qquad c = c^{\mathrm{MB}}, \label{rho a b first and third kernel} \\
& C = r^{rc}C^{\mathrm{MB}}. \label{bigC first and third kernel}
\end{align}
The relations \eqref{rho a b first and third kernel} were already verified in \cite[Remark 5.2]{CLMMuttalib}, and we now verify that \eqref{bigC first and third kernel} holds. Let us explicitly write the dependence of $C^{\text{MB}}$ on $\theta$ and $\alpha$. From \cite[Theorem 1.1]{CLMMuttalib}, we have
\begin{align}
C^{\text{MB}}( \theta,\alpha )&=\frac{G(1+\alpha)}{(2\pi)^{\frac{\alpha}{2}}} \exp \big( d(1,\alpha) - d(\theta,\alpha) \big)  \exp\left( \frac{24 \alpha (\alpha +2)+15+3\theta + 4 \theta^{2}}{24(1+\theta)} \ln \theta \right) \nonumber
\\
& \quad \times \exp \left( \frac{6\alpha \theta - 6 \alpha (1+\alpha)-(\theta-1)^{2}}{12 \theta} \ln(1+\theta) \right), \label{MB big C}
\end{align}
where $d(\theta,\alpha)$ is a regularized sum, see \cite[Eq. (1.12)]{CLMMuttalib}. If $\theta$ is rational, it follows from \cite[Proposition 1.4]{CLMMuttalib} that $d(\theta,\alpha)$ can be expressed in terms of $\zeta'(-1)$ and Barnes' $G$-function evaluated at certain points. By specializing \cite[Proposition 1.4]{CLMMuttalib} for $\theta = \frac{1}{r}$, we obtain
\begin{align}
d(\tfrac{1}{r},\alpha) &=r\zeta'(-1) + \frac{1+(1+2\alpha)r}{4} \ln(2\pi)- \frac{1}{12}\left( 3+\frac{1}{r}+r+6\alpha(1+r + \alpha r) \right) \ln r \nonumber
\\ 
& \quad - \sum_{k = 1}^{r} \ln G \left( 1+\alpha + \frac{k}{r} \right). \label{d with theta=1/r}
\end{align}
Substituting \eqref{d with theta=1/r} into \eqref{MB big C} with $\theta=\frac{1}{r}$, we obtain after a direct computation that
\begin{align}
r^{rc}C^{\text{MB}}( \tfrac{1}{r},\alpha ) = & \; \frac{G(1+\alpha)\prod_{j=1}^{r}G\left(1+\alpha + \frac{k}{r}\right)}{G(2+\alpha)(2\pi)^{\frac{r-1+2r\alpha}{4}}} \exp\{ (1-r)\zeta'(-1)\} \nonumber \\
& \times \exp \left\{ \frac{1-r(5+12\alpha)+2r^{2}(1+6\alpha+6\alpha^{2})}{24(1+r)}\ln r \right\} \nonumber \\
& \times \exp \left\{ - \frac{1-2r(1+3\alpha)+r^{2}(1+6\alpha + 6 \alpha^{2})}{12 r}\ln(1+r) \right\}. \label{lol1}
\end{align}
On the other hand, by substituting the particular values of $\nu_{j}$ given by \eqref{cond parameters for relation first and third kernel} into \eqref{constant C2}, another long but straightforward computation shows that
\begin{align*}
C\big( (r;\alpha,\alpha+\tfrac{1}{r},\ldots,\alpha+\tfrac{r-1}{r}), (0,-) \big)
\end{align*}
is also given by the right-hand-side of \eqref{lol1}, which proves \eqref{bigC first and third kernel}.

\paragraph{Poles-zeros cancellation.} If one increases simultaneously $r$ and $q$ by $1$, with $\nu_{r+1}$ and $\mu_{q+1}$ such that $\nu_{r+1} = \mu_{q+1}$, it is easy to see from \eqref{def of F} that $F$ (and henceforth the kernel $\mathbb{K}$) remains unchanged. We verify directly from \eqref{constant C2} that
\begin{align*}
C\big( (r+1;\nu_{1},\ldots,\nu_{r},\nu_{r+1}),(q+1;\mu_{1},\ldots,\mu_{q},\nu_{r+1}) \big) = C\big( (r;\nu_{1},\ldots,\nu_{r}),(q;\mu_{1},\ldots,\mu_{q}) \big),
\end{align*}
which is consistent with this observation.

\paragraph{Numerical confirmations.} 
\begin{figure}[h]
\begin{tikzpicture}[master]
\node at (0,0) {};
\node at (0,-0.1) {\includegraphics[scale=0.35]{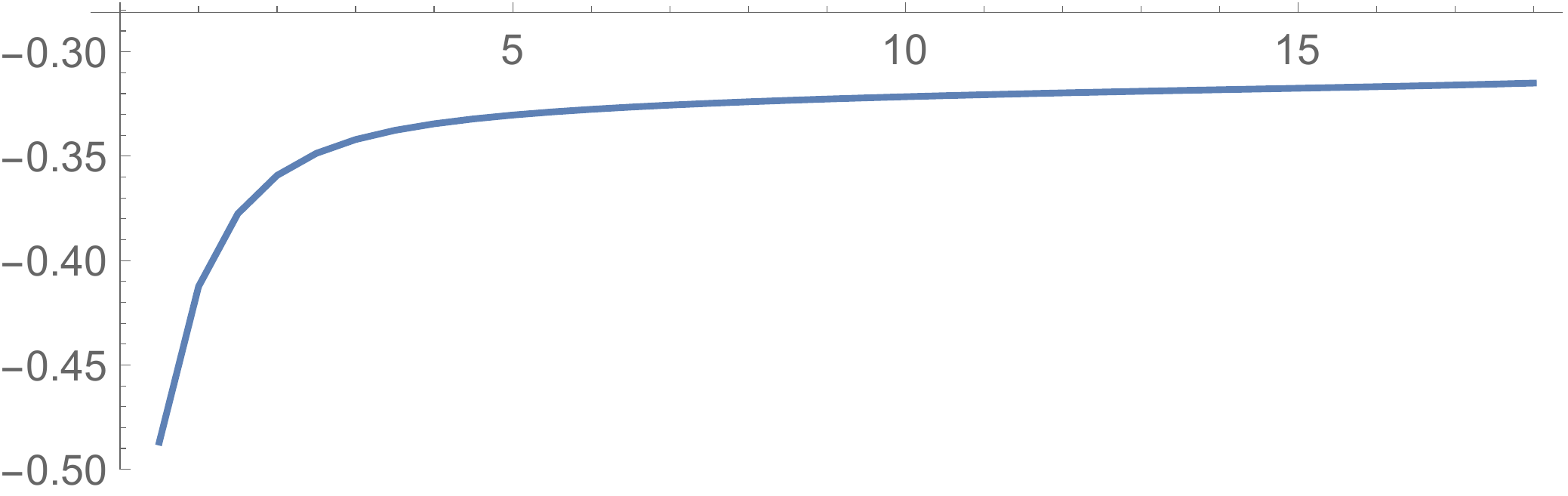}};
\end{tikzpicture}
\begin{tikzpicture}[slave]
\node at (0,0) {\includegraphics[scale=0.35]{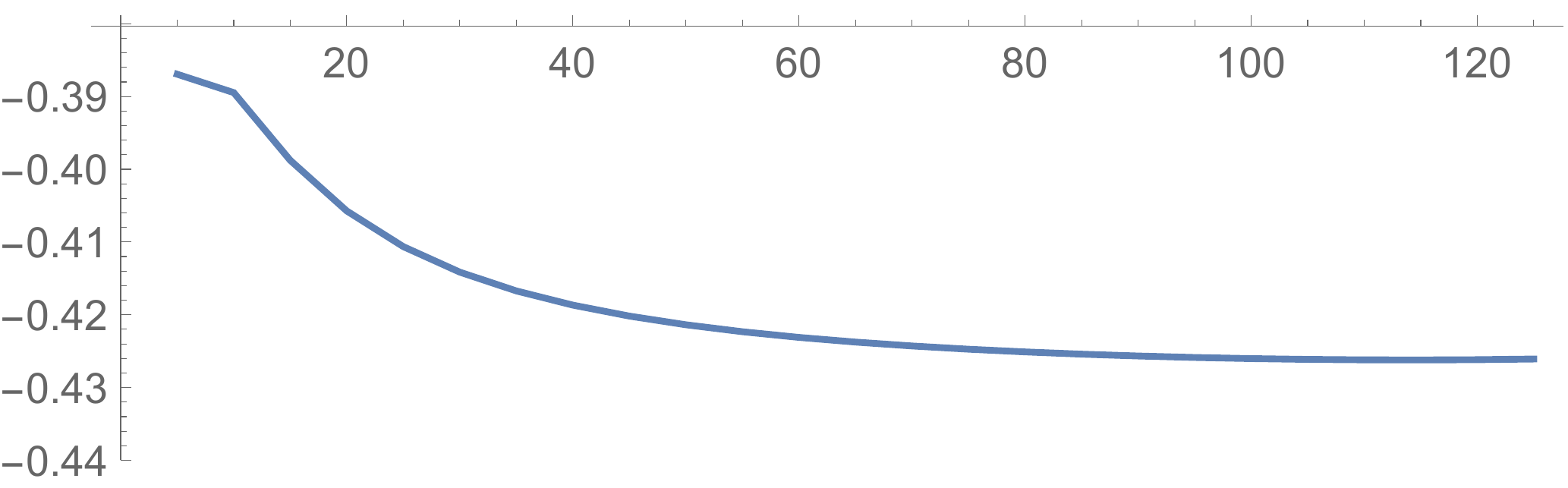}};
\end{tikzpicture}
\caption[]{\label{fig: num check}  Numerical confirmations of Theorem \ref{mainthm}.
}
\end{figure}
The all-order expansion \eqref{all order expansion} implies in particular that the 
function 
\begin{align}\label{function in num check}
s \mapsto s^{\rho}\Big(\ln \det \big( 1- \mathbb{K}|_{[0,s]}\big) -\Big[as^{2\rho}+bs^{\rho} + c \ln s + \ln C \Big] \Big)
\end{align} 
converges to a constant $C_{1}$ as $s \to + \infty$. Figure \ref{fig: num check} shows the graph of the function \eqref{function in num check} in the following two cases:
\begin{align*}
\mbox{Left: } \quad r=3, \quad q=2, \quad \nu_{1}=1.31, \quad \nu_{2} = 2.15, \quad \nu_{3} = 3.19, \quad \mu_{1}=1.87, \quad \mu_{2} = 2.61, \\
\mbox{Right: } \quad r=4, \quad q=1, \quad \nu_{1}=1.31, \quad \nu_{2} = 2.15, \quad \nu_{3} = 2.61, \quad \nu_{4} = 3.19, \quad \mu_{1}=1.87, 
\end{align*}
which correspond to the coefficients
\begin{align*}
& \mbox{Left: } & & \rho=\frac{1}{2}, & & a = 1, & & b=4.34, & & c\approx -1.551, & & \ln C \approx -2.963, \\
& \mbox{Right: } & & \rho=\frac{1}{4}, & & a = \frac{4}{\sqrt{3}}, & & b\approx 12.97, & & c\approx -2.437, & & \ln C \approx -10.097.
\end{align*}
For these two choices of the parameters, we observe in Figure \ref{fig: num check} that the function \eqref{function in num check} seems indeed to converge to a constant as $s \to + \infty$. These observations support the validity of \eqref{higher order constants11}, \eqref{higher order constants12}, \eqref{higher order constants2} and Theorem \ref{mainthm}, and even slightly more than that, since they also suggest that the first error term in the large $s$ asymptotics of $\ln \det \big( 1- \mathbb{K}|_{[0,s]})$ is of order $s^{-\rho}$, which is consistent with \eqref{all order expansion}. These numerical checks have been made using the linear algebra package for Fredholm determinants of Bornemann \cite{Bornemann}.

\subsection{Outline of the proof}\label{outlinesubsec}
The expressions \eqref{higher order constants11}, \eqref{higher order constants12} and \eqref{higher order constants2} for the coefficients $\rho$, $a$, $b$ and $c$, as well as the all-order expansion \eqref{all order expansion}, were obtained in \cite{ClaeysGirSti,CLMMuttalib} via a method that we briefly explain here. There is a standard procedure, named after Its, Izergin, Korepin and Slavnov (IIKS) \cite{IIKS1990}, which expresses the logarithmic derivative\footnote{The derivative can be taken with respect to any given parameter of the associated kernel, as long as the kernel depends smoothly on this parameter.} of the Fredholm determinant of an integrable kernel in terms of the solution of a Riemann--Hilbert (RH) problem. The kernel $\mathbb{K}$ given in \eqref{def kernels} is not integrable (in general), but it was shown in \cite{ClaeysGirSti} that the IIKS procedure can still be applied (this fact is far from obvious---it is based on ideas from \cite{BC2012b,BC2012} and uses the Mellin transform, see \cite{ClaeysGirSti} for details). Using the IIKS procedure, the authors of \cite{ClaeysGirSti} were able to express
\begin{align}
\partial_s \ln \det \big( 1- \mathbb{K}|_{[0,s]}\big) \label{lol2}
\end{align}
in terms of the solution $Y$ of a $2\times 2$ matrix Riemann--Hilbert (RH) problem. A Deift/Zhou \cite{DZ1993} nonlinear steepest descent analysis of this RH problem yields an all-order expansion of $Y$ as $s \to + \infty$  \cite{ClaeysGirSti, CLMMuttalib}, and hence also of \eqref{lol2}. The constants $\rho$, $a$, $b$ and $c$ of \eqref{higher order constants11}, \eqref{higher order constants12} and \eqref{higher order constants2} were then obtained after integrating in $s$ the large $s$ asymptotics of \eqref{lol2}. However, with this method, $C$ appears as an integration constant, and therefore remains undetermined.

\medskip The considerations above show that the determination of $C$ is an essentially different problem than finding the other constants in the large $s$ asymptotics of $\det ( 1- \mathbb{K}|_{[0,s]})$. The general strategy of the present work consists of using different, and more complicated, differential identities than \eqref{lol2}. Similar strategies were used in some of the works mentioned below (\ref{higher order constants2}), and also in \cite{ErcMcL, BleIts} (although in \cite{ErcMcL, BleIts} the multiplicative constants were not determined). 

\medskip As mentioned, we need to use different differential identities than \eqref{lol2} to evaluate $C$ (i.e. differential identities with respect to other parameters than $s$). The large gap asymptotics for the Bessel point process is known up to and including the constant (see \eqref{Asymptotics Bessel}), so the idea is to find a path in the set of parameters $r$, $q$, $\nu_{1},\ldots,\nu_{r}$, $\mu_{1},\ldots,\mu_{q}$ which interpolates smoothly between the Bessel kernel and $\mathbb{K}$. The existence of such a path is a priori not clear, since the parameters $r$ and $q$ are integers. The simple, but central idea of this paper is to first set the parameters $\nu_{1},\ldots,\nu_{r},\mu_{1},\ldots,\mu_{q}$ associated to $\mathbb{K}$ equal to $\nu_{\min}$, where
\begin{align*}
\nu_{\min}=\min \{ \nu_1,\ldots,\nu_r,\mu_{1},\ldots,\mu_{q} \},
\end{align*}
and ``smooth" the product of Gamma functions in (\ref{def of F}) by considering the following kernel:
\begin{align} \label{def of Kr}
\mathbb{K}_r(x,y) = \int_{\gamma}\frac{du}{2\pi i} \int_{\tilde{\gamma}}\frac{dv}{2\pi i} \frac{F_{r}(u)}{F_{r}(v)} \frac{x^{-u}y^{v-1}}{v-u},
\end{align}
where $r\geq 1$ and $\nu > -1$ are real-valued parameters and
\begin{align} \label{def of Fr}
F_r(z)=\frac{\Gamma(z)}{\Gamma(1+\nu-z)^r},
\end{align}
and where we choose the branch such that $F_{r}(z)$ is analytic for $z \in \mathbb{C}\setminus (-\infty,0]$. Starting with $r = 1$ (note that $\mathbb{K}_r$ reduces to the Bessel kernel for $r=1$), we first increase $r$ from $1$ to $r-q$ and then successively move each of the remaining parameters from $\nu_{\min}$ to its desired value. The process relies on the successive integration of appropriate differential identities for the following quantities:
\begin{align}
& \partial_r \ln \det \big( 1- \mathbb{K}_r|_{[0,s]}\big), \label{first diff identity outline} \\
& \partial_{\nu_\ell} \ln \det \big( 1- \mathbb{K}|_{[0,s]}\big), \quad \ell \in \{1,\ldots,r\}, \label{second diff identity outline} \\
& \partial_{\mu_\ell} \ln \det \big( 1- \mathbb{K}|_{[0,s]}\big), \quad \ell \in \{1,\ldots,q\}. \label{third diff identity outline}
\end{align}
More precisely, let us define $\mathbb{K}^{(\ell)}$, $\ell \in \{0,1,\ldots,r, r+1,\ldots, r+q \}$ by
\begin{align}
\mathbb{K}^{(\ell)} = \begin{cases}
\mathbb{K}|_{\nu_{\ell+1}=...=\nu_{r}=\mu_{1}=...=\mu_{q} = \nu_{\min}}, & \mbox{if } \ell \in \{0,...,r-1\}, \\
\mathbb{K}|_{\mu_{\ell-r+1}=...=\mu_{q} = \nu_{\min}}, & \mbox{if } \ell \in \{r,...r+q\}.
\end{cases} \label{def of K ell}
\end{align}
Note that $\mathbb{K}^{(0)} = \mathbb{K}_{r-q}$, where $\mathbb{K}_{r-q}$ is defined by \eqref{def of Kr} (with $r$ replaced by $r-q$ and $\nu$ by $\nu_{\min}$), and $\mathbb{K}^{(r+q)} = \mathbb{K}$. By integrating successively \eqref{first diff identity outline}, \eqref{second diff identity outline} and \eqref{third diff identity outline}, we obtain
\begin{align}
& \ln \det \big( 1- \mathbb{K}_{r-q}|_{[0,s]}\big) = \ln \det \big( 1- \mathbb{K}_{r=1}|_{[0,s]}\big) + \int_{1}^{r-q} \partial_{r'} \ln \det \big( 1- \mathbb{K}_{r'}|_{[0,s]}\big) dr', \label{first integration outline} \\
& \ln \det \big( 1- \mathbb{K}^{(\ell)}|_{[0,s]}\big) = \ln \det \big( 1- \mathbb{K}^{(\ell-1)}|_{[0,s]}\big) + \int_{\nu_{\min}}^{\nu_{\ell}} \partial_{\nu_{\ell}'} \ln \det \big( 1- \mathbb{K}^{(\ell)}|_{[0,s]}\big) d\nu_{\ell}', \quad \ell = 1,...,r, \nonumber \\
& \ln \det \big( 1- \mathbb{K}^{(r+\ell)}|_{[0,s]}\big) = \ln \det \big( 1- \mathbb{K}^{(r+\ell-1)}|_{[0,s]}\big) + \int_{\nu_{\min}}^{\mu_{\ell}} \partial_{\mu_{\ell}'} \ln \det \big( 1- \mathbb{K}^{(r+\ell)}|_{[0,s]}\big) d\mu_{\ell}', \quad \ell = 1,...,q. \nonumber
\end{align}
Since the large $s$ asymptotics of $\ln \det \big( 1- \mathbb{K}_{r=1}|_{[0,s]}\big)$ are known up to and including the constant term, see \eqref{Asymptotics Bessel}, this method allows us to obtain $C$ by keeping track of the term of order $1$ in the identities \eqref{first integration outline}. 

\begin{remark}\label{focusremark}
In this work, we focus on proving the expression \eqref{constant C2} for $C$. We could also have computed the coefficients $\rho$, $a$, $b$ and $c$ with the same method (this would have provided an alternative proof and another consistency check for these constants), but these constants are already known \cite{ClaeysGirSti,CLMMuttalib}, and in order to limit the complexity and length of the paper, we have decided to not pursue this direction.
\end{remark}

\subsection{Organization of the paper}
By employing the IIKS procedure, we will express the quantities \eqref{first diff identity outline}, \eqref{second diff identity outline} and \eqref{third diff identity outline} in terms of $Y$. 
Since we use the same RH problem as in \cite{ClaeysGirSti, CLMMuttalib}, we can recycle some of the analysis of these papers. We present the necessary material from \cite{ClaeysGirSti, CLMMuttalib} in Section \ref{Section: preliminaries}. In Section \ref{Section: differential identities}, we express the quantities \eqref{first diff identity outline}, \eqref{second diff identity outline} and \eqref{third diff identity outline} in terms of $Y$. Section \ref{Section: Asymptotics r} is devoted to the first differential identity \eqref{first diff identity outline}. More precisely, we compute the large $s$ asymptotics of \eqref{first diff identity outline} up to and including the constant term and then perform the integration with respect to $r'$ in \eqref{first integration outline}. In Section \ref{Section: Asymptotics nu mu}, we proceed similarly with the differential identities with respect to $\nu_{i}$, $i=1,...,r$ and $\mu_{j}$, $j=1,...q$. The proof of Theorem \ref{mainthm} is completed in Section \ref{Section: proof main thm}.


\subsection*{Acknowledgements}
Support is acknowledged from the European Research Council, Grant Agreement No. 682537, the Swedish Research Council, Grant No. 2015-05430, the G\"oran Gustafsson Foundation, and the Ruth and Nils-Erik Stenb\"ack Foundation.

\section{Background from \cite{ClaeysGirSti,CLMMuttalib}} \label{Section: preliminaries}
In this section we recall some results from \cite{ClaeysGirSti,CLMMuttalib} that will be used throughout the paper. The notation adopted in this paper is the same as in \cite{CLMMuttalib}, and is also almost identical to the notation used in \cite{ClaeysGirSti}. The only difference with \cite{ClaeysGirSti} is that the function denoted by $\mathcal{G}$ in this paper and in \cite{CLMMuttalib} is instead denoted by $G$ in \cite{ClaeysGirSti}. In this paper, as well as in \cite{CLMMuttalib}, $G$ denotes Barnes' $G$-function.\footnote{If $G$ has subscripts and superscripts, such as $G_{p,q}^{m,n}$, then it denotes the Meijer $G$-function.} 

We start by stating the RH problem for $Y$. 

\subsubsection*{RH problem for $Y$}
\begin{itemize}
\item[(a)] $Y: \C \setminus (\gamma \cup \tilde{\gamma}) \to \mathbb{C}^{2 \times 2}$ is analytic, where $\gamma$ and $\tilde{\gamma}$ are the contours shown in Figure \ref{fig: gamma and gammatilde}.
\item[(b)] $Y$ has continuous boundary values $Y_+$ and $Y_-$ on $\gamma \cup \tilde{\gamma}$ from the left ($+$) side and right ($-$) side of $\gamma \cup \tilde{\gamma}$, respectively, and obeys the jump relations
\begin{align}
Y_+(z)&=Y_-(z) \begin{pmatrix}
1 & -s^{-z}F(z) \\ 0 & 1
\end{pmatrix},& & \mbox{ if $z \in \gamma$}, \label{jumps for Y on gamma}
\\
Y_+(z)&=Y_-(z) \begin{pmatrix}
1 &  0 \\ s^{z}F(z)^{-1} & 1
\end{pmatrix},& & \mbox{ if $z \in \tilde{\gamma}$}, \label{jumps for Y on gamma tilde}
\end{align}
where $F$ is defined as in \eqref{def of F}.
\item[(c)] $Y$ admits an expansion of the form
\begin{align}\label{eq: expansion of Y}
Y(z) = I +\frac{Y_1(s)}{z} + \bigO\big( z^{-2} \big), \qquad \mbox{as } z \to \infty,
\end{align}
where $Y_1$ is a $2 \times 2$ matrix that depends on $s$ but not on $z$.
\end{itemize}
It is shown in \cite{ClaeysGirSti} that the solution of the RH problem for $Y$ exists and is unique for all $s \in \R$. In the steepest descent analysis of the RH problem, the authors of \cite{ClaeysGirSti} introduce a sequence of transformations $Y\mapsto U \mapsto T \mapsto S \mapsto R$, where $R$ is the solution of a small norm RH problem. We only recall here what is needed for the proof of Theorem \ref{mainthm}, and refer to \cite{ClaeysGirSti, CLMMuttalib} for more details. 
Let us choose the branch for $\ln F$ such that
\begin{align}\label{def of ln F}
\ln F(z)= \ln \Gamma(z) + \sum_{k=1}^{q} \ln \Gamma(1+\mu_k-z) - \sum_{j=1}^{r} \ln \Gamma(1+\nu_j-z),
\end{align}
and the branches on the right-hand-side of \eqref{def of ln F} are the principal ones. Then $z \mapsto \ln F(z)$ is analytic for $z \in \mathbb{C}\setminus \big( (-\infty,0] \cup [1+\nu_{\min},+\infty) \big)$. The first transformation $Y \mapsto U$ involves the change of variables
\begin{align*}
z(\zeta)=is^\rho \zeta + \frac{1+\nu_{\min}}{2},
\end{align*}
where we recall that $\nu_{\min}=\min \{ \nu_1,\ldots,\nu_r,\mu_{1},\ldots,\mu_{q} \}$. The asymptotics of $\ln F(z(\zeta))$ as $s^\rho \zeta \to \infty$ are 
\begin{align*}
\ln F\bigg(is^\rho \zeta + \frac{1+\nu_{\min}}{2} \bigg) &= is^\rho \zeta \ln(s) + is^\rho \big(c_1 \zeta \ln(i\zeta) + c_2 \zeta \ln (-i\zeta) +c_3 \zeta  \big)
\\
& \quad + c_4 \ln (s) + c_5 \ln (i\zeta)+c_6 \ln (-i\zeta) + c_7 +\frac{c_8}{i s^\rho \zeta} + \bigO \bigg( \frac{1}{s^{2\rho}\zeta^2} \bigg),
\end{align*}
where the constants $c_1,\ldots,c_8$ are given by
\begin{align} \label{eq: values of constants cj}
& c_1=1, \qquad  c_2 =r-q, \qquad c_3=-(r-q+1) \qquad c_4=\frac{\nu_{\min}}{2}+\frac{1}{1+r-q}\Bigg(\sum_{k=1}^{q} \mu_k - \sum_{j=1}^{r} \nu_j \Bigg), \nonumber \\
& c_5=\frac{\nu_{\min}}{2}, \qquad c_6 =(r-q) \frac{\nu_{\min}}{2} + \sum_{k=1}^{q} \mu_k - \sum_{j=1}^{r} \nu_j, \qquad c_7=\frac{1+q-r}{2} \ln(2\pi), \nonumber
\\
 & c_8 = \frac{1+r-q}{8} \bigg( \nu_{\min}^2- \frac{1}{3}\bigg) -\frac{1}{2}\Bigg(\sum_{k=1}^{q} \mu_k^2 - \sum_{j=1}^{r} \nu_j^2 \Bigg) +\frac{\nu_{\min}}{2} \Bigg(\sum_{k=1}^{q} \mu_k - \sum_{j=1}^{r} \nu_j \Bigg).
\end{align}
We define the function $\mathcal{G}$ by
\begin{align} \label{def of G}
\mathcal{G}(\zeta)=F\bigg( is^\rho \zeta + \frac{1+\nu_{\min}}{2} \bigg) e^{-is^\rho(\zeta \ln s - h(\zeta) )},
\end{align}
where $h(\zeta)= -\zeta (c_1 \ln (i\zeta)+c_2 \ln(-i\zeta)+c_3)$. Note that the function $\mathcal{G}(\zeta)$ also depends on the parameters $s$, $r$, $q$, $\nu_j$, and $\mu_k$, although this is not indicated in the notation. 
The $T \mapsto S$ transformation utilizes a $g$-function $\zeta \mapsto g(\zeta)$ which is analytic for $\zeta \in \mathbb{C}\setminus \Sigma_{5}$, where $\Sigma_{5}$ is the union of two segments
\begin{align}
\Sigma_5 = [b_1,0] \cup [0,b_2],
\end{align}
oriented from left to right, see Figure \ref{fig: Sigma5}. The points $b_{1}$ and $b_{2}$ which characterize $\Sigma_{5}$ are defined by
\begin{align} \label{b1b2def}
b_2 =-\overline{b_1} = |b_2| e^{i\phi}, \quad \phi \in \big[0,\tfrac{\pi}{2} \big),
\end{align}
where
\begin{align*}
\re b_2 &= -\re b_1 = 2\bigg(\frac{c_2}{c_1}\bigg)^{\frac{c_1-c_2}{2(c_1+c_2)}} e^{-\frac{c_1+c_2+c_3}{c_1+c_2}}= 2(r-q)^{\frac{1-r+q}{2(1+r-q)}},
\\
\sin \phi &= \frac{c_2-c_1}{c_2+c_1}=\frac{r-q-1}{r-q+1} \in [0,1).
\end{align*}
The $g$-function is defined via its second derivative given by
\begin{align} \label{def of g''}
g''(\zeta)=-i\frac{c_1+c_2}{2} \bigg( \frac{1}{\zeta} -\frac{1}{r(\zeta)}+ \frac{i \im b_2}{\zeta r(\zeta)} \bigg),
\end{align}
where $\zeta \mapsto r(\zeta)$ is analytic for $\zeta \in \mathbb{C}\setminus \Sigma_{5}$ and is defined by
\begin{align}
r(\zeta)= \sqrt{(\zeta-b_1)(\zeta-b_2)},
\end{align}
where the branch is fixed such that $r(\zeta) \sim \zeta$ as $\zeta \to \infty$. For $\zeta \in \Sigma_{5}$, one has $r_{+}(\zeta) + r_{-}(\zeta) = 0$. It can be shown that $g''(\zeta) = \bigO(\zeta^{-3})$ as $\zeta \to \infty$, and then $g'$ and $g$ are defined by
\begin{align}
g'(\zeta) = \int_{\infty}^{\zeta} g''(\xi) d\xi, \qquad g(\zeta) = \int_{\infty}^{\zeta} g'(\xi) d\xi,
\end{align}
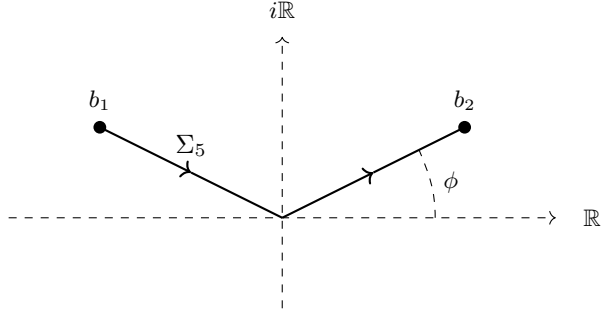
\begin{figure}
	\begin{center}
		\begin{tikzpicture}[scale=1.2]
		
		\draw[dashed,->] (5,-1) -- (5,2);
		\draw[dashed,->] (2,0) -- (8,0);	
		\node at (5,2.3) {\small $i \mathbb{R}$};
		\node at (8.4,0) {\small $\mathbb{R}$};
		
		\draw[thick,->] (3,1) -- (4,0.5);
		\draw[thick] (4,0.5) -- (5,0);	
		\draw[thick,->] (5,0) -- (6,0.5);	
		\draw[thick] (6,0.5) -- (7,1);	
		\node at (4,0.8) {$\Sigma_5$};
		
		\draw[dashed] (6.677,0) to [out=90,in=-67.5] (6.5,0.75);
		\node at (6.85,0.4) {\small $\phi$};	
		
		\fill (3,1) circle (0.07cm);
		\node at (3,1.3) {\small $b_1$};
		\fill (7,1) circle (0.07cm);
		\node at (7,1.3) {\small $b_2$};		
		\end{tikzpicture}
		\caption{The points $b_2= |b_2|e^{i\phi}$ and $b_1=-\overline{b_2}$ and the contour $\Sigma_5=[b_1,0]\cup [0,b_2]$.} \label{fig: Sigma5}
	\end{center}
\end{figure}
where the path of integration lies in $\C \setminus \Sigma_5$. The RH problem for $S$ has exponentially decaying jumps outside $\Sigma_{5}$, and one needs to construct approximations to $S$ in different regions of the complex plane. Let $\mathbb{D}_{\delta}(b_{1})$ and $\mathbb{D}_{\delta}(b_{2})$ denote two disks of sufficiently small radius $\delta > 0$ centered at $b_{1}$ and $b_{2}$, respectively.  For $\zeta \in \mathbb{C}\setminus \big( \mathbb{D}_{\delta}(b_{1}) \cup \mathbb{D}_{\delta}(b_{2}) \big)$, $S$ is approximated by a so-called global parametrix $P^{(\infty)}$, while for $\zeta \in \mathbb{D}_{\delta}(b_{1}) \cup \mathbb{D}_{\delta}(b_{2})$, $S$ is approximated by local parametrices $P$ that are defined in terms of Airy functions. We omit the exact definition of $P$ here, the interested reader can find it in \cite{ClaeysGirSti,CLMMuttalib}. The construction of $P^{(\infty)}$ is given in terms of a function $p$ that will be important for us and which is defined by
\begin{align}\label{def of p}
p(\zeta) &= -\frac{r(\zeta)}{2 \pi i} \int_{\Sigma_5} \frac{\ln \mathcal{G}(\xi)}{r_+(\xi)} \frac{d\xi }{\xi-\zeta},
\end{align}
where the branch of $\ln \mathcal{G}$ is such that
\begin{align*}
\ln \mathcal{G}(\zeta) =\ln F\bigg( is^\rho \zeta + \frac{1+\nu_{\min}}{2} \bigg) -is^\rho(\zeta \ln s - h(\zeta) ).
\end{align*}
The function $p$ has the following jumps across $\Sigma_5$:
\begin{align} \label{preliminary jumps}
p_+(\zeta) + p_-(\zeta) = - \ln \mathcal{G}(\zeta), \qquad \zeta \in \Sigma_5,
\end{align}
and the following asymptotics as $\zeta \to \infty$:
\begin{align*}
& p(\zeta) = p_0 + \bigO (\zeta^{-1}), \qquad \mbox{ where } \qquad  p_0 = \frac{1}{2 \pi i} \int_{\Sigma_5} \frac{\ln \mathcal{G}(\mathcal{\xi})}{r_+(\xi)} d\xi.
\end{align*}
The global parametrix $P^{(\infty)}(\zeta)$ is defined by \cite[Eq. (3.51)]{ClaeysGirSti}
\begin{align} \label{def of Pinfty}
P^{(\infty)}(\zeta)= e^{-p_0 \sigma_3} Q^{(\infty)}(\zeta) e^{p(\zeta) \sigma_3} \quad \text{with}\quad   Q^{(\infty)}(\zeta) = \frac{1}{2} \begin{pmatrix}
\frac{\gamma(\zeta)+\gamma(\zeta)^{-1}}{2} & \frac{\gamma(\zeta)-\gamma(\zeta)^{-1}}{2i}
\\
\frac{\gamma(\zeta)-\gamma(\zeta)^{-1}}{-2i} & \frac{\gamma(\zeta)+\gamma(\zeta)^{-1}}{2}
\end{pmatrix},
\end{align}
where the branch of the function
\begin{align*}
\gamma(\zeta)=\bigg( \frac{\zeta-b_1}{\zeta-b_2} \bigg)^{\frac{1}{4}}
\end{align*}
is chosen such that $\gamma(\zeta)$ is analytic on $\C \setminus \Sigma_5$ and $\gamma(\zeta) \sim 1$ as $\zeta \to \infty$.
The solution of the RH problem for $R$ is then given by \cite{ClaeysGirSti} (see also \cite[Eq. (2.32)]{CLMMuttalib})
\begin{align}\label{def of R}
R(\zeta) = e^{p_0\sigma_3}S(\zeta) \times \begin{cases} P(\zeta)^{-1} e^{-p_0 \sigma_3}, \quad &\zeta \in \mathbb{D}_{\delta}(b_1) \cup\mathbb{D}_{\delta}(b_2),  \\
P^{(\infty)}(\zeta)^{-1} e^{-p_0 \sigma_3},\quad &\zeta \in \Gamma_R\setminus(\mathbb{D}_{\delta}(b_1) \cup\mathbb{D}_{\delta}(b_2)).
\end{cases}
\end{align}
Let $\{\Sigma_j\}_1^4$ denote the contours
\begin{align*}
\Sigma_2=-\overline{\Sigma_1}= b_2 + e^{i(\phi+\epsilon)} \R_{\ge 0}, \qquad \Sigma_4=-\overline{\Sigma_3}= b_2 + e^{-i\epsilon} \R_{\ge 0}
\end{align*}
for some fixed $\epsilon \in (0,\frac{\pi}{10})$, with the orientation from left to right, and define
\begin{align} \label{def of tilde Sigmaj}
\tilde{\Sigma}_j = \Sigma_j \setminus (\mathbb{D}_\delta (b_1)\cup \mathbb{D}_\delta (b_2)), \qquad j=1,...,5.
\end{align}
The function $R$ defined in \eqref{def of R} is analytic for $\zeta \in \mathbb{C}\setminus \Gamma_{R}$, where
\begin{align*}
\Gamma_R = \big(\partial \mathbb{D}_\delta (b_1)\cup \partial \mathbb{D}_\delta (b_2)\big) \cup \bigcup_{j=1}^5 \tilde{\Sigma}_j,
\end{align*}
and $\partial \mathbb{D}_\delta (b_1)$ and $\partial \mathbb{D}_\delta (b_2)$ are oriented clockwise, see Figure \ref{fig: GammaR}.
\begin{figure}
	\begin{center}
		\begin{tikzpicture}[scale=1.2]
		
		\draw[dashed,->] (5,-1) -- (5,2);
		\draw[dashed,->] (2,0) -- (8,0);	
		\node at (5,2.3) {\small $i \mathbb{R}$};
		\node at (8.4,0) {\small $\mathbb{R}$};
		
		\draw[thick,->] (3,1) -- (4,0.5);
		\draw[thick] (4,0.5) -- (5,0);	
		\draw[thick,->] (5,0) -- (6,0.5);	
		\draw[thick] (6,0.5) -- (7,1);

		\draw[thick] (3,1) -- (1.5,0.75);
		\draw[thick,->] (0,0.5) -- (1.5,0.75);	
		\draw[thick,->] (7,1) -- (8.5,0.75);	
		\draw[thick] (8.5,0.75) -- (10,0.5);

		\draw[thick] (3,1) -- (1.5,2.5);
		\draw[thick,->] (0.5,3.5) --  (1.5,2.5);	
		\draw[thick,->] (7,1) -- (8.5,2.5);	
		\draw[thick] (8.5,2.5) -- (9.5,3.5);

		\draw[thick, fill= white, 
		decoration={markings, mark=at position 0.125 with {\arrow{<}}},
		postaction={decorate}
		] (3,1) circle (0.5cm);
		\draw[thick, fill= white, 
		decoration={markings, mark=at position 0.4 with {\arrow{<}}},
		postaction={decorate}] (7,1) circle (0.5cm);
		
		\fill (3,1) circle (0.07cm);
		\node at (3,1.2) {\small $b_1$};
		\fill (7,1) circle (0.07cm);
		\node at (7,1.2) {\small $b_2$};		
		\end{tikzpicture}
		\caption{The contour $\Gamma_R$.} \label{fig: GammaR}
	\end{center}
\end{figure}
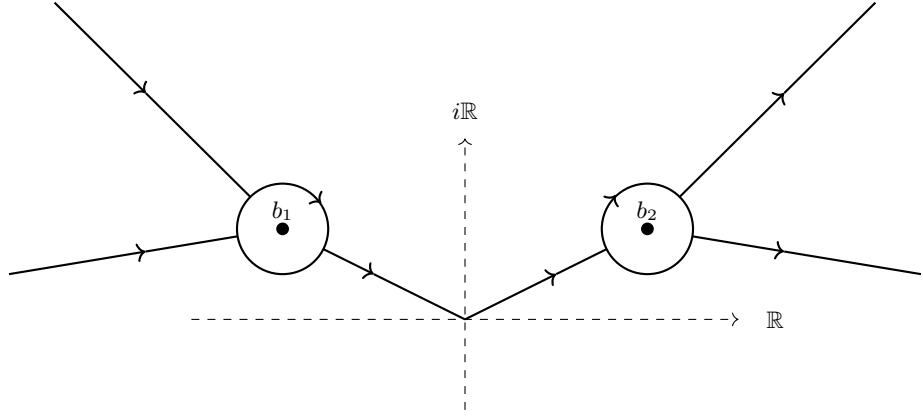
\begin{remark}\label{remark: the case Fr}
The jumps for $Y$ depend on $F$. If $F$ is replaced by $F_{r}$ in \eqref{jumps for Y on gamma}-\eqref{jumps for Y on gamma tilde}, then the steepest descent analysis of $Y$ has not been carried out in \cite{ClaeysGirSti}. However, it is not hard to see that the same analysis applies also in this case (with $\nu_{\min} = \nu$); the only difference is that the coefficients \eqref{eq: values of constants cj} are replaced by
\begin{align}
& c_1=1, \qquad c_2 =r, \qquad c_3=-(r+1), \qquad c_4=\nu \bigg(\frac{1}{2}-\frac{r}{1+r}\bigg), \nonumber \\
& c_5=\frac{\nu}{2}, \qquad c_6=-\frac{r\nu}{2}, \qquad c_7=\frac{1-r}{2} \ln(2\pi), \qquad c_8= \frac{1+r}{8} \bigg( \nu^2- \frac{1}{3}\bigg).  \label{eq: values of constants cj Fr}
\end{align}
\end{remark}

\section{Differential identities in $r$, $\nu_\ell$, and $\mu_\ell$} \label{Section: differential identities}
In this section, we express the logarithmic derivatives \eqref{first diff identity outline}, \eqref{second diff identity outline} and \eqref{third diff identity outline} in terms of the RH problem for $Y$ via the IIKS procedure. By combining \cite[Propositions 2.1 and 2.2]{ClaeysGirSti} and \cite[Theorem 2.1]{BC2012}, we obtain
\begin{align}
\partial_r \ln \det  \big( 1- \mathbb{K}_r|_{[0,s]}\big) &= \int_{\gamma \cup \tilde{\gamma}} \tr [Y_-^{-1}(z)Y_-'(z) \partial_r J(z) J^{-1}(z)] \frac{dz}{2\pi i}, \label{lol3}
\\
\partial_{\nu_{\ell}} \ln \det  \big( 1- \mathbb{K}|_{[0,s]}\big) &= \int_{\gamma \cup \tilde{\gamma}} \tr [Y_-^{-1}(z)Y_-'(z) \partial_{\nu_{\ell}} J(z) J^{-1}(z)] \frac{dz}{2\pi i}, \quad \ell \in \{1,\ldots,r\}, \label{lol4}
\\
\partial_{\mu_{\ell}} \ln \det  \big( 1- \mathbb{K}|_{[0,s]}\big) &= \int_{\gamma \cup \tilde{\gamma}} \tr [Y_-^{-1}(z)Y_-'(z) \partial_{\mu_{\ell}} J(z) J^{-1}(z)] \frac{dz}{2\pi i}, \quad \ell \in \{1,\ldots,q\}, \label{lol5}
\end{align}
where the RH solution $Y$ in \eqref{lol3} has the jumps \eqref{jumps for Y on gamma}-\eqref{jumps for Y on gamma tilde} with $F$ replaced by $F_{r}$ (defined in \eqref{def of Fr}). The quantity $(\partial_{r}J)J^{-1}$ in \eqref{lol3} is given by
\begin{align}
\partial_r J(z) J^{-1}(z) = \ln \Gamma (1+\nu -z)\big( J(z)-I \big) \sigma_3.
\end{align}
In \eqref{lol4}-\eqref{lol5}, $Y$ satisfies the jumps \eqref{jumps for Y on gamma}-\eqref{jumps for Y on gamma tilde} with $F$ given by \eqref{def of F}, and we have
\begin{align}
\partial_{\nu_{\ell}} J(z) J^{-1}(z) &= \psi (1+\nu_{\ell}-z)\big( J(z)-I \big) \sigma_3, & \ell &\in \{ 1,\ldots, r \},
\\
\partial_{\mu_{\ell}} J(z) J^{-1}(z) &= -\psi (1+\mu_{\ell}-z)\big( J(z)-I \big) \sigma_3, & \ell &\in \{ 1,\ldots, q \},
\end{align}
where $\psi=(\ln \Gamma)'$ denotes the di-gamma function. The same arguments as in the proof of \cite[Lemma 6.1]{CLMMuttalib} apply here (so we do not provide details), and we obtain
\begin{align} \label{first diff identity r}
\partial_r \ln \det  \big( 1- \mathbb{K}_r|_{[0,s]}\big) &= \frac{1}{2}\int_{\gamma \cup \tilde{\gamma}}\ln \Gamma (1+\nu -z) \tr [Y_+^{-1}(z)Y_+'(z) \sigma_3 -Y_-^{-1}(z)Y_-'(z)\sigma_3] \frac{dz}{2\pi i}, 
\\ \label{first diff identity nu}
\partial_{\nu_{\ell}} \ln \det  \big( 1- \mathbb{K}|_{[0,s]}\big) &=\frac{1}{2} \int_{\gamma \cup \tilde{\gamma}} \psi (1+\nu_{\ell}-z)  \tr [Y_+^{-1}(z)Y_+'(z) \sigma_3 -Y_-^{-1}(z)Y_-'(z)\sigma_3]\frac{dz}{2\pi i},
\\ \label{first diff identity mu}
\partial_{\mu_{\ell}} \ln \det  \big( 1- \mathbb{K}|_{[0,s]}\big) &=-\frac{1}{2} \int_{\gamma \cup \tilde{\gamma}}\psi (1+\mu_{\ell}-z)  \tr [Y_+^{-1}(z)Y_+'(z) \sigma_3 -Y_-^{-1}(z)Y_-'(z)\sigma_3] \frac{dz}{2\pi i}, 
\end{align}
where $\ell \in \{ 1,\ldots, r \}$ in \eqref{first diff identity nu} and $\ell \in \{ 1,\ldots, q \}$ in \eqref{first diff identity mu}. 

Using the chain of transformations $Y \mapsto U \mapsto T \mapsto S \mapsto R$ in the steepest descent analysis of \cite{ClaeysGirSti}, we rewrite the differential identities \eqref{first diff identity r}, \eqref{first diff identity nu}, and \eqref{first diff identity mu} in a way that is more convenient for the asymptotic analysis as $s \to +\infty$. Recall that the transformation $Y \mapsto U$ involves the change of variables $z=i \zeta s^\rho + (1+\nu_{\min})/2$ (and that $\nu_{\min} = \nu$ in \eqref{first diff identity r}, see also Remark \ref{remark: the case Fr}). The functions
\begin{align} \label{functions with poles}
\ln \Gamma\bigg( \frac{1+\nu}{2}-i \zeta s^\rho  \bigg), \quad \psi \bigg( \frac{1+2\nu_{\ell}-\nu_{\min}}{2}-i \zeta s^\rho  \bigg), \quad \text{and} \quad  \psi \bigg( \frac{1+2\mu_{\ell}-\nu_{\min}}{2}-i \zeta s^\rho  \bigg)
\end{align}
appearing in \eqref{first diff identity r}, \eqref{first diff identity nu} and \eqref{first diff identity mu} have infinitely many poles on $i \mathbb{R}^{-}$. These poles depend on $s$ and approach $0$ as $s \to + \infty$. For example, the left-most function in \eqref{functions with poles} has simple poles at $\{\zeta_{j}\}_{j=0}^{+\infty} \subset i \mathbb{R}^{-}$, where
\begin{align*}
\zeta_j=\frac{1}{i s^{\rho}} \bigg(\frac{1+\nu}{2}+j \bigg), \qquad j=0,1,...
\end{align*} 
Following \cite{CLMMuttalib}, we define for $K > |b_2|$ the contour $\sigma_K$ as given in Figure \ref{fig: sigma}. The contour $\sigma_K$ surrounds $\Sigma_5$ in the positive direction in such a way that the poles of \eqref{functions with poles} lie in the region exterior to $\sigma_{K}$. The circular part of $\sigma_{K}$ has radius $K$. We choose the contour $\sigma_K$ to cross the imaginary axis at the point $\zeta_0/2$ and to have a horizontal part of constant length as $s$ changes. Note that, since the poles of \eqref{functions with poles} approach $0$ as $s \to + \infty$, the contour $\sigma_{K}$ depends on $s$, even if $K$ is independent of $s$. We define $\sigma = \sigma_{2|b_2|}$. Furthermore, we define the contour
\begin{align*}
\tilde{\Sigma}_K = \bigcup_{j=1}^4 \tilde{\Sigma}_j \setminus \{|\zeta| \leq K\},
\end{align*}
where the contours $\tilde{\Sigma}_j$ are defined by \eqref{def of tilde Sigmaj}.
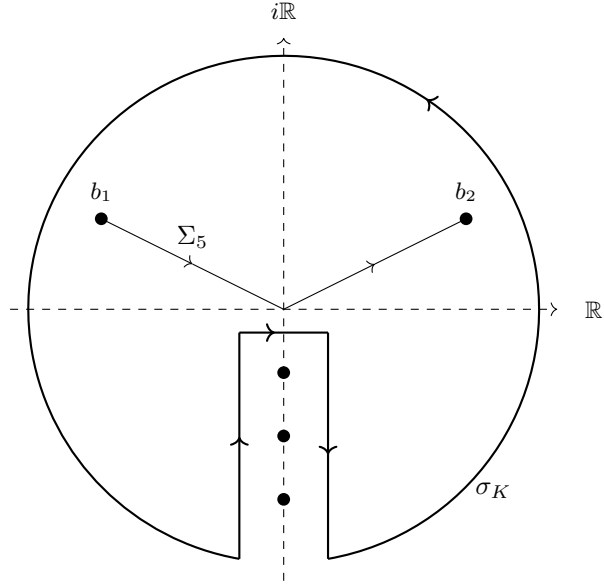
\begin{figure}
	\begin{center}
		\begin{tikzpicture}[scale=1.2]
		
		\draw[dashed,->] (5,-3) -- (5,3);
		\draw[dashed,->] (2,0) -- (8,0);	
		\node at (5,3.3) {\small $i \mathbb{R}$};
		\node at (8.4,0) {\small $\mathbb{R}$};
		
		\draw[thick, decoration={markings, mark=at position 0.4 with {\arrow{>}}}	,postaction={decorate}] ([shift=(-80:2.8cm)]5,0) arc (-80:260:2.8cm);
		\draw[thick,decoration={markings, mark=at position 0.5 with {\arrow{<}}}	,postaction={decorate}] ([shift=(-80:2.8cm)]5,0) to ([shift=(-80:2.8cm)]5,2.5);
		\draw[thick,decoration={markings, mark=at position 0.7 with {\arrow{<}}}	,postaction={decorate}] ([shift=(-80:2.8cm)]5,2.5) to ([shift=(260:2.8cm)]5,2.5);
		\draw[thick, decoration={markings, mark=at position 0.5 with {\arrow{<}}}	,postaction={decorate}] ([shift=(260:2.8cm)]5,2.5) to ([shift=(260:2.8cm)]5,0);
		\node at (7.3,-2) {$\sigma_K$};
		
		\fill (5,-0.7) circle (0.07cm);
		\fill (5,-1.4) circle (0.07cm);
		\fill (5,-2.1) circle (0.07cm);
		
		\draw[->] (3,1) -- (4,0.5);
		\draw (4,0.5) -- (5,0);	
		\draw[->] (5,0) -- (6,0.5);	
		\draw (6,0.5) -- (7,1);	
		\node at (4,0.8) {$\Sigma_5$};
		
		\fill (3,1) circle (0.07cm);
		\node at (3,1.3) {\small $b_1$};
		\fill (7,1) circle (0.07cm);
		\node at (7,1.3) {\small $b_2$};		
		\end{tikzpicture}
		\caption{The contour $\sigma_K$ surrounds $\Sigma_5$ but does not enclose any of the poles of the functions in \eqref{functions with poles}.} \label{fig: sigma}
	\end{center}
\end{figure}

\newpage
\begin{lemma}[Differential identities]\label{lemma: differential identities} Let $K$ be such that $K > 2 |b_2|$. Then the following statements hold:
\begin{itemize}
\item[(a)] Let $r \geq 1$ and $\nu=\nu_{\min} > -1$. Then
\begin{align} \label{diff identity r}
\partial_r \ln \det  \big( 1- \mathbb{K}_r|_{[0,s]}\big) &= I_{1,r}+I_{2,r}+I_{3,r}(K)+I_{4,r}(K),
\end{align}
where
\begin{align}
 I_{1,r} &= - s^{\rho}  \int_{\sigma} \ln \Gamma\Big(\frac{1+\nu}{2} -i s^{\rho} \zeta\Big) \, g^{\prime}(\zeta) \frac{d\zeta}{2\pi i}, \label{def of I1r}
\\
 I_{2,r} &= - \frac{1}{2} \int_{\sigma} \ln \Gamma\Big(\frac{1+\nu}{2} -i s^{\rho} \zeta\Big) \tr\Big[P^{(\infty)}(\zeta)^{-1}P^{(\infty)}(\zeta)^{\prime}\sigma_{3}\Big]\frac{d\zeta}{2\pi i},  \label{def of I2r}
\\
 I_{3,r}(K) &= - \frac{1}{2} \int_{\sigma_K} \ln \Gamma\Big(\frac{1+\nu}{2} -i s^{\rho} \zeta\Big) \tr \Big[P^{(\infty)}(\zeta)^{-1}e^{-p_0 \sigma_3}R^{-1}(\zeta)R^{\prime}(\zeta)e^{p_0 \sigma_3}P^{(\infty)}(\zeta) \sigma_{3} \Big] \frac{d\zeta}{2\pi i}, \label{def of I3r}
\\ \nonumber
 I_{4,r}(K) &=  \frac{1}{2} \int_{\tilde{\Sigma}_K} \ln \Gamma\Big(\frac{1+\nu}{2} -i s^{\rho} \zeta\Big) \times
 \\
 &\qquad  \tr\Big[P^{(\infty)}(\zeta)^{-1}e^{-p_0 \sigma_3}\big( R_+^{-1}(\zeta)R_+^{\prime}(\zeta) -R_-^{-1}(\zeta)R_-^{\prime}(\zeta)\big)e^{p_0 \sigma_3}P^{(\infty)}(\zeta) \sigma_{3} \Big] \frac{d\zeta}{2\pi i}. \label{def of I4r}
\end{align}
\item[(b)] Let $r>q\ge0$ be integers. Let $\ell \in \{1,\ldots, r \}$ and $k \in \{1,\ldots, q\}$. Then
\begin{align} 
\partial_{\nu_{\ell}}  \ln \det  \big( 1- \mathbb{K}|_{[0,s]}\big)&= I_{1,\nu_{\ell}}+I_{2,\nu_{\ell}}+I_{3,\nu_{\ell}}(K)+I_{4,\nu_{\ell}}(K), \label{diff identity munu}
\\
\partial_{\mu_k}  \ln \det  \big( 1- \mathbb{K}|_{[0,s]}\big)&= -I_{1,\mu_k}-I_{2,\mu_k}-I_{3,\mu_k}(K)-I_{4,\mu_k}(K), \label{diff identity munu2}
\end{align}
where, for $\alpha \in \{\nu_{\ell}, \mu_k \}$,
\begin{align}\label{I1alphadef}
I_{1,\alpha} = & - s^{\rho}  \int_{\sigma} \psi \Big(\frac{1+2\alpha-\nu_{\min}}{2} -i s^{\rho} \zeta\Big) \, g^{\prime}(\zeta) \frac{d\zeta}{2\pi i}, 
\\ \label{I2alphadef}
I_{2,\alpha} =& - \frac{1}{2} \int_{\sigma} \psi \Big(\frac{1+2\alpha-\nu_{\min}}{2} -i s^{\rho} \zeta\Big) \tr\Big[P^{(\infty)}(\zeta)^{-1}P^{(\infty)}(\zeta)^{\prime}\sigma_{3}\Big]\frac{d\zeta}{2\pi i}, 
	\\ \nonumber
I_{3,\alpha}(K)=& - \frac{1}{2} \int_{\sigma_K} \psi\Big(\frac{1+2\alpha-\nu_{\min}}{2} -i s^{\rho} \zeta\Big) 
	\\ \label{I3alphadef}
& \times \tr \Big[P^{(\infty)}(\zeta)^{-1}e^{-p_0 \sigma_3}R^{-1}(\zeta)R^{\prime}(\zeta)e^{p_0 \sigma_3}P^{(\infty)}(\zeta) \sigma_{3} \Big] \frac{d\zeta}{2\pi i},
	\\\nonumber
I_{4,\alpha}(K) = &\; \frac{1}{2} \int_{\tilde{\Sigma}_K} \psi\Big(\frac{1+2\alpha-\nu_{\min}}{2} -i s^{\rho} \zeta\Big) 
\tr\Big[P^{(\infty)}(\zeta)^{-1}e^{-p_0 \sigma_3}
	\\ \label{I4alphadef}
&  \times \big( R_+^{-1}(\zeta)R_+^{\prime}(\zeta) -R_-^{-1}(\zeta)R_-^{\prime}(\zeta)\big)e^{p_0 \sigma_3}P^{(\infty)}(\zeta) \sigma_{3} \Big] \frac{d\zeta}{2\pi i}.
\end{align}
\end{itemize}

\end{lemma}

\begin{proof}
The proof is analogous to the proof of \cite[Lemma 6.2]{CLMMuttalib} and consists of implementing the chain of transformations $Y\mapsto U \mapsto T \mapsto S \mapsto R$ in \eqref{first diff identity r}, \eqref{first diff identity nu}, and \eqref{first diff identity mu}, and performing a contour deformation.
\end{proof}

Sections \ref{Section: Asymptotics r} and \ref{Section: Asymptotics nu mu} are devoted to the computation of the constant terms in the large $s$ asymptotics of the right-hand sides of \eqref{diff identity r}, \eqref{diff identity munu}--\eqref{diff identity munu2} and to the integration of these identities. The proof of Theorem \ref{mainthm} is then given in Section \ref{Section: proof main thm}.

\section{Asymptotics of the differential identity in $r$} \label{Section: Asymptotics r}
In this section, we compute the large $s$ asymptotics of the four quantities $I_{1,r}$, $I_{2,r}$, $I_{3,r}(K)$, and $I_{4,r}(K)$ appearing on the right-hand side of the differential identity \eqref{diff identity r} in $r$. By integrating the resulting asymptotics with respect to $r$, we obtain the constant term in the large $s$ asymptotics of $\ln \det  \big( 1- \mathbb{K}_r|_{[0,s]}\big)$.

Throughout this section, we assume that $r\geq 1$ and $\nu > -1$. The quantities $c_1,\ldots,c_8$ and $b_1,b_2$ are defined by \eqref{eq: values of constants cj Fr} and \eqref{b1b2def}, respectively.
As mentioned in Remark \ref{focusremark}, we focus in this work on proving the expression \eqref{constant C2} for $C$, because the coefficients $\rho$, $a$, $b$ and $c$ are already known from \cite{ClaeysGirSti,CLMMuttalib}. Therefore, to avoid unnecessary computations, we introduce the notation $\Omega$.

\begin{notation} 
Let $t \in\R$ and $f,g:(t,\infty) \to \C$. The notation
\begin{align*}
f(s)=\Omega\big( g(s) \big), \qquad \mbox{as } s \to +\infty,
\end{align*}
means that either $f \equiv 0$ or that there exist $c > 0$ and $s_{0}>0$ independent of $s$ such that 
\begin{align*}
|f(s)| \geq c |g(s)|, \qquad \mbox{for all } s \geq s_{0}.
\end{align*}
\end{notation}

\subsection{Asymptotics of $I_{1,r}$}
\begin{proposition} \label{prop: I_1r}
Let $\nu > -1$ and let $I_{1,r}$ be the function defined by \eqref{diff identity r}. Then
\begin{align} \label{asymtptotics I_1r}
I_{1,r}= \Omega(\ln(s^{\rho}))+ I_{1,r}^{(c)} + \mathcal{O}\bigg( \frac{\ln ( s^\rho )}{s^\rho} \bigg)
\end{align}
as $s\to +\infty$ uniformly for $r$ in compact subsets of $[1,+\infty)$, where
\begin{align} \label{I_1r3 explicit}
I_{1,r}^{(c)}&=\frac{c_1+c_2}{2}\bigg\{ \frac{1-3\nu^2}{24} + \zeta'(-1)-\ln G \bigg(\frac{\nu+1}{2}\bigg) +\frac{\nu-1}{2} \ln \Gamma\bigg(\frac{\nu+1}{2}\bigg)\bigg\} \bigg( 1- \frac{\im b_2}{|b_2|}\bigg) \nonumber
\\
& \quad + \frac{c_1+c_2}{48}(1-3\nu^2)  \bigg( \left( 1+\frac{\im b_{2}}{|b_{2}|} \right)\ln\left( \frac{|b_2|+\im b_2}{2} \right) - 2\frac{\im b_2}{|b_{2}|} \ln |b_{2}| \bigg).
\end{align}
\end{proposition}
\begin{proof}
We define the function
\begin{align}\label{def of Psi primitive}
\Psi(\zeta)= s^\rho \int_0^\zeta \ln \Gamma \bigg( \frac{1+\nu}{2}-i s^\rho \xi \bigg)d\xi.
\end{align}
Then $\Psi(\zeta)$ is analytic on $\sigma$ and an integration by parts and \eqref{def of g''} yield
\begin{align*}
I_{1,r}&=   \int_{\sigma}\Psi(\zeta) g^{\prime \prime}(\zeta) \frac{d\zeta}{2\pi i} = i\frac{c_1+c_2}{2}\int_{\sigma}\Psi(\zeta) \bigg( 1- i\frac{\im b_2}{\zeta} \bigg) \frac{1}{r(\zeta)}\frac{d\zeta}{2\pi i},
\end{align*}
where we have used the fact that $\frac{1}{\zeta}\Psi(\zeta)$ has no pole at $\zeta = 0$. We first collapse the contour $\sigma$ onto $\Sigma_{5}$. Second, recalling that $r_{+}(\zeta) + r_{-}(\zeta) = 0$ for $\zeta \in \Sigma_{5}$, we rewrite the resulting integral only in terms of $r_{+}$. Finally, we deform the contour on the $+$ side of $\Sigma_{5}$ into another contour $\gamma_{b_2b_1}$, which is the part of the counterclockwise oriented circle with radius $|b_2|$ centered at the origin going from $b_2$ to $b_1$. We note again that, since $\frac{1}{\zeta}\Psi(\zeta)$ is analytic at $\zeta = 0$, there is no residue at $\zeta = 0$ during this contour deformation. This gives
\begin{align}\label{new expression I_1r}
I_{1,r}&= i(c_1+c_2)  \int_{\gamma_{b_2b_1}}\Psi(\zeta) \bigg( 1- i\frac{\im b_2}{\zeta} \bigg) \frac{1}{r(\zeta)}\frac{d\zeta}{2\pi i}.
\end{align}
It remains to compute the large $s$ asymptotics of $\Psi(\zeta)$ uniformly for $\zeta \in \gamma_{b_{2}b_{1}}$. The following formula is useful for us (cf. \cite[Eq. 5.17.4]{NIST})
\begin{align}  \label{identity int ln Gamma}
\int_1^{z} \ln \Gamma(z') dz' = \frac{z-1}{2} \ln (2\pi) - \frac{(z-1)z}{2} +(z-1) \ln \Gamma(z) - \ln G(z),
\end{align}
where $G$ is Barnes' $G$-function. By applying \eqref{identity int ln Gamma} twice with
\begin{align*}
z=\frac{1+\nu}{2}-is^{\rho}\zeta \qquad \mbox{ and } \qquad z=\frac{1+\nu}{2}
\end{align*}
in \eqref{def of Psi primitive}, we obtain
\begin{align} \label{Psi explicit}
\Psi(\zeta)=& \; i \int_{\frac{1+\nu}{2}}^{\frac{1+\nu}{2}-is^{\rho} \zeta} \ln \Gamma(z')dz' \nonumber \\
= & \frac{i}{2}\bigg[s^{2\rho}\zeta^2  +is^\rho \zeta \big(\nu -\ln(2\pi)\big)-2 \ln  \frac{G\big(\frac{1+\nu}{2}-is^{\rho}\zeta\big)}{G\big( \frac{1+\nu}{2} \big)} -(1-\nu)\ln \frac{\Gamma\big(\frac{1+\nu}{2}-is^{\rho}\zeta\big)}{\Gamma\big( \frac{1+\nu}{2} \big)} \nonumber
\\
& \quad -2is^\rho \zeta \ln\Gamma\bigg(\frac{1+\nu}{2}-is^{\rho}\zeta\bigg) \bigg].
\end{align}
The large $z$ asymptotics of $\ln \Gamma(z)$ and $\ln G(z)$ are given by (cf. \cite[Eqs. 5.11.1 and 5.17.5]{NIST})
\begin{align} \label{asymptotics lnGamma and lnG}
\ln G(z+1)&=\frac{z^2}{4}+ z\ln \Gamma(z+1)-\bigg( \frac{z(z+1)}{2} +\frac{1}{12}\bigg) \ln z -\frac{1}{12}+\zeta'(-1) + O(z^{-2}), \nonumber
\\
\ln \Gamma(z)&= (z-\tfrac{1}{2})\ln z -z+\frac{1}{2} \ln(2\pi)+\frac{1}{12z} +O(z^{-3})
\end{align}
as $z\to \infty$ uniformly for $|\arg z| < \pi-\epsilon$ for some $\epsilon >0$. This implies
\begin{align} \nonumber
\Psi(\zeta)&= \Omega\big( \ln(s^\rho) \big) -\frac{i}{24}\bigg\{ 1-3 \nu^2 -\frac{i \pi}{2} +24 \zeta'(-1)+\frac{3 i  \pi \nu^2}{2} - 24 \ln G \bigg(\frac{\nu+1}{2}\bigg)
\\
&\quad +12(\nu-1)\ln \Gamma \bigg(\frac{\nu+1}{2}\bigg) + (1-3\nu^2)\ln \zeta  \bigg\}+ \mathcal{O}\big( s^{-\rho} \big) \label{asymptotics of Psi}
\end{align}
as $s\to + \infty$ uniformly for $\zeta \in \gamma_{b_2 b_1}$ and $r$ in compact subsets of $(0,+\infty)$. Substituting \eqref{asymptotics of Psi} into \eqref{new expression I_1r} gives \eqref{asymtptotics I_1r} and, in particular,
\begin{align*}
I_{1,r}^{(c)}&= (c_1+c_2)\bigg\{ \frac{1-3\nu^2}{24}\bigg(1-\frac{\pi i}{2}\bigg) + \zeta'(-1)-\ln G \bigg(\frac{\nu+1}{2}\bigg) +\frac{\nu-1}{2} \ln \Gamma\bigg(\frac{\nu+1}{2}\bigg)  \bigg\} 
\\
& \quad\times \int_{\gamma_{b_2b_1}} \bigg(1-\frac{i \im b_2}{\zeta }  \bigg) \frac{1}{r(\zeta)} \frac{d\zeta}{2\pi i}
 + \frac{c_1+c_2}{24}(1-3\nu^2)  \int_{\gamma_{b_2b_1}} \ln(\zeta) \bigg(1-\frac{i \im b_2}{\zeta }  \bigg) \frac{1}{r(\zeta)} \frac{d\zeta}{2\pi i}.
\end{align*}
It was shown in \cite[Lemma 7.2]{CLMMuttalib} that
\begin{align}\label{integral identities I1r} \nonumber
2\int_{\gamma_{b_2b_1}}  \frac{1}{r(\zeta)} \frac{d\zeta}{2\pi i}&=1, & 2\int_{\gamma_{b_2b_1}}  \frac{\ln (\zeta)}{r(\zeta)} \frac{d\zeta}{2\pi i}&= \ln\big( i(|b_2|+\im b_2) \big)-\ln(2),
\\
2\int_{\gamma_{b_2b_1}}  \frac{1}{\zeta r(\zeta)} \frac{d\zeta}{2\pi i}&=-\frac{i}{|b_2|}, & 2\int_{\gamma_{b_2b_1}}  \frac{\ln(\zeta)}{\zeta r(\zeta)} \frac{d\zeta}{2\pi i}&= \frac{\ln \big(\frac{2i |b_2|^2}{|b_2|+\im b_2} \big)}{i|b_2|},
\end{align}
which proves \eqref{I_1r3 explicit} and the proposition.
\end{proof}

\subsection{Asymptotics of $I_{2,r}$}\label{subsection: asymp for I2r}
In this subsection we compute large $s$ asymptotics for $I_{2,r}$, which we recall is given by
\begin{align*}
I_{2,r}= - \frac{1}{2} \int_{\sigma} \ln \Gamma\Big(\frac{1+\nu}{2} -i s^{\rho} \zeta\Big) \tr\Big[P^{(\infty)}(\zeta)^{-1}P^{(\infty)}(\zeta)^{\prime}\sigma_{3}\Big]\frac{d\zeta}{2\pi i}.
\end{align*}
Recalling the definition \eqref{def of Pinfty} of the global parametrix $P^{(\infty)}(\zeta)$, a straightforward calculation yields
\begin{align}\label{lol10}
\tr\Big[P^{(\infty)}(\zeta)^{-1}P^{(\infty)}(\zeta)^{\prime}\sigma_{3}\Big] = \tr \Big[ Q^{(\infty)}(\zeta)^{-1}Q^{(\infty) \prime}(\zeta) \sigma_3 \Big] + \tr [p'(\zeta)I]=2 p'(\zeta).
\end{align}
Therefore, integrating by parts, using the jump condition \eqref{preliminary jumps} of $p(\zeta)$, and then collapsing the contour $\sigma$ onto $\Sigma_{5}$, we obtain
\begin{align*}
I_{2,r}=  -is^\rho  \int_{\sigma} \psi \Big(\frac{1+\nu}{2} -i s^{\rho} \zeta\Big) p(\zeta)\frac{d\zeta}{2\pi i} = is^\rho  \int_{\Sigma_5} \psi \Big(\frac{1+\nu}{2} -i s^{\rho} \zeta\Big) \big(p_+(\zeta)-p_-(\zeta)\big)\frac{d\zeta}{2\pi i} =Z_r+X_r,
\end{align*}
where
\begin{align} \label{def of Zr and Xr}
Z_r&= -2is^\rho \int_{\gamma_{b_2b_1}}\psi \Big(\frac{1+\nu}{2} -i s^{\rho} \zeta\Big) p(\zeta)\frac{d\zeta}{2\pi i}, \nonumber
\\
X_r&=is^\rho \int_{\Sigma_5}\psi \Big(\frac{1+\nu}{2} -i s^{\rho} \zeta\Big) \ln \mathcal{G}(\zeta)\frac{d\zeta}{2\pi i}.
\end{align}
It remains to find the large $s$ asymptotics of $Z_r$ and $X_r$.

\subsubsection{Asymptotics of $X_r$}
The large $s$ asymptotics of $X_{r}$ are described in terms of the Hurwitz zeta function $\zeta(z,u)$ which is defined for $\re z > 1$ and $u \neq 0, -1, -2, \dots$ by
$$\zeta(z,u) = \sum_{n=0}^\infty \frac{1}{(n + u)^z}.$$
\begin{proposition}\label{prop: asymptotics Xr}
Let $\nu \ge 0$ and $X_r$ be defined by \eqref{def of Zr and Xr}. Then
\begin{align}
X_{r} &=\Omega \big( \ln (s^\rho) \big) +X_{1,r}^{(c)}+X_{2,r}^{(c)}+X_{3,r}^{(c)}+ \mathcal{O}\bigg( \frac{\ln ( s^\rho )}{s^\rho} \bigg) \label{asymp for Xr Prop}
\end{align}
as $s\to +\infty$ uniformly for $r$ in compact subsets of $[1,+\infty)$, where
\begin{align}
X_{1,r}^{(c)} = &\; \frac{i}{48 \pi} \bigg\{ 6r\nu^{2} \big[ \ln^{2}(-i b_{1})-\ln^{2}(-ib_{2}) \big] + 6 \nu \ln(2\pi) \big[ \ln(ib_{2})-\ln(ib_{1}) \big] \nonumber \\
& + \big[ \ln(-ib_{2})-\ln(-ib_{1})\big] \Big( (1+r)(1-3\nu^{2})+6\nu (1-2r)\ln(2\pi) \Big) \nonumber \\
& + 6\nu^{2} \Big( \ln(-i b_{2}) \ln(ib_{2})-\ln(-ib_{1})\ln(ib_{1}) \Big) \bigg\}, \label{def of X1rc}
\\
X_{2,r}^{(c)}= &\; \ln G(1+\nu)- \zeta'(-1)-\frac{\nu\ln(2\pi)}{4}+ \frac{3\nu^2-1}{24}-\zeta'\bigg( -1, \frac{1+\nu}{2}+1\bigg) +\frac{\nu+1}{2}\ln \bigg( \frac{\nu+1}{2}\bigg),\label{def of X2rc}
\\
X_{3,r}^{(c)}= &\;\frac{1+r+3(r-3)\nu^2}{24} \frac{\ln(b_1/b_2)\ln(-b_1b_2)}{4\pi i} 
- \frac{(1-3\nu^2)(1+r)-6(r-1)\nu \ln(2\pi)}{24}  \frac{\ln(b_1/b_2)}{2 \pi i}, \label{def of X3rc}
\end{align}
where $\zeta'(z,u) = \partial_{z} \zeta(z,u)$ denotes the derivative in the $z$-variable of the Hurwitz zeta function.
\end{proposition}

\begin{proof}
We first rewrite the integral $X_r$ in a convenient way. Performing an integration by parts yields
\begin{align} \label{Xr after int by parts}
X_r&=-\frac{\ln \Gamma \big( \frac{1+\nu}{2}-is^\rho\zeta\big)  \ln \mathcal{G}(\zeta) }{2\pi i} \Bigg|_{\zeta=b_1}^{b_2}+\int_{\Sigma_5} \ln \Gamma  \bigg( \frac{1+\nu}{2}-is^\rho \zeta\bigg) \frac{\mathcal{G}'(\zeta)}{\mathcal{G}(\zeta)} \frac{d\zeta}{2\pi i}.
\end{align}
By the definition \eqref{def of G} of $\mathcal{G}(\zeta)$ (with $F$ replaced by $F_{r}$) and the identity $\rho^{-1} = c_{1}+c_{2}$, it holds that
\begin{align} \label{derivative of lnG}
\frac{\mathcal{G}'(\zeta)}{\mathcal{G}(\zeta)} = is^\rho \bigg\{ \psi\bigg( \frac{1+\nu}{2}+is^\rho \bigg)-\ln(is^\rho \zeta)  + r \psi\bigg( \frac{1+\nu}{2}-is^\rho \bigg)- r\ln(-is^\rho \zeta)  \bigg\}.
\end{align}
By substituting \eqref{derivative of lnG} into \eqref{Xr after int by parts} and using the change of variables $w=i s^\rho \zeta$, we split $X_{r}$ as
\begin{align} \label{splitting Xr}
X_r=X_{1,r}+X_{2,r}+X_{3,r},
\end{align}
where $X_{1,r}$, $X_{2,r}$ and $X_{3,r}$ are given by
\begin{align*}
X_{1,r}&= -\frac{\ln \Gamma \big( \frac{1+\nu}{2}-is^\rho\zeta\big) \ln \mathcal{G}(\zeta)}{2\pi i}\Bigg|_{\zeta=b_1}^{b_2},
\\
X_{2,r}&=  \int_{is^\rho\Sigma_5} \ln \Gamma \bigg( \frac{1+\nu}{2} -w \bigg)  \bigg\{ \psi\bigg( \frac{1+\nu}{2}+w\bigg)-\ln(w) -f(w)\bigg\} \frac{dw}{2\pi i},
\\
X_{3,r}&= \int_{is^\rho\Sigma_5} \ln \Gamma \bigg( \frac{1+\nu}{2} -w \bigg)  \bigg\{ r \psi\bigg( \frac{1+\nu}{2}-w \bigg)-r\ln(-w)+f(w)\bigg\} \frac{dw}{2\pi i},
\end{align*}
with
\begin{align} \label{def of f}
f(w)=\frac{\nu}{2} \bigg( \frac{1}{w-m} -\frac{m}{(w-m)^2} \bigg) +\frac{1-3\nu^2}{24(w-m)^2} \qquad \mbox{ and } \qquad m =\frac{1+\nu}{2}.
\end{align}
We have added and subtracted the term $f(w)$ in order to make the integrand of $X_{2,r}$ vanish as $w^{-2} \ln w $ as $w\to \infty$. Indeed, from \eqref{asymptotics lnGamma and lnG} and the asymptotic formula \cite[Eq 5.11.2]{NIST} of the di-gamma function given by
\begin{align} \label{asymtptoics for psi}
\psi(z) = \ln z - \frac{1}{2z} -\frac{1}{12z^2}+ \bigO \bigg( \frac{1}{z^4} \bigg), \qquad z \to \infty,
\end{align}
for $|\arg z | < \pi-\delta$ with some fixed $\delta > 0$, we have
\begin{align} \label{asymptotics of psi with w}
\psi\bigg( \frac{1+\nu}{2}+w\bigg) &= \ln w +\frac{\nu}{2w} +\frac{1-3\nu^2}{24 w^2} + \bigO \bigg( \frac{1}{w^3} \bigg), 
\\
&= \ln w  +f(w)+\bigO \bigg( \frac{1}{w^3} \bigg), \qquad w \to \infty, \label{asymptotics of psi with w and f}
\end{align}
where $|\arg w| < \pi-\delta$ with some fixed $\delta > 0$. Note that the integral $X_{2,r}$ is convergent as long as $m \notin is^{\rho}\Sigma_{5}$; the choice $m=\tfrac{1+\nu}{2}$ is made because it makes the upcoming computations easier. The remainder of the proof consists of computing the large $s$ asymptotics of $X_{1,r}$, $X_{2,r}$ and $X_{3,r}$.

\paragraph{Asymptotics of $X_{1,r}$.} From \cite[Eq. (3.15)]{ClaeysGirSti}, we have
\begin{align} \label{asymptotics G}
\ln \mathcal{G}(\zeta) = c_4 \ln s +c_5 \ln (i\zeta)+c_6 \ln (-i\zeta) + c_7 + \frac{c_8}{is^\rho \zeta} + \bigO \bigg( \frac{1}{s^{2\rho}\zeta^{2}}\bigg),\qquad s^{\rho}\zeta \to \infty.
\end{align}
The asymptotic formula \eqref{asymptotics G} is in particular valid for $\zeta = b_{1}$ and $\zeta = b_{2}$. By combining these asymptotics together with \eqref{asymptotics lnGamma and lnG} and \eqref{eq: values of constants cj Fr}, we obtain
\begin{align} \label{asymptotics of X1r}
X_{1,r}=\Omega \big( \ln (s^\rho) \big) +X_{1,r}^{(c)}+ \mathcal{O}\bigg( \frac{\ln ( s^\rho )}{s^\rho} \bigg) \qquad \mbox{as } s \to + \infty,
\end{align}
where $X_{1,r}^{(c)}$ is given by \eqref{def of X1rc}.

\paragraph{Asymptotics of $X_{2,r}$.} Recall that $X_{2,r}$ is given by
\begin{align*}
X_{2,r}&=  \int_{is^\rho\Sigma_5} \ln \Gamma \bigg( \frac{1+\nu}{2} -w \bigg)  \bigg\{ \psi\bigg( \frac{1+\nu}{2}+w\bigg)-\ln(w) -f(w)\bigg\} \frac{dw}{2\pi i},
\end{align*}
and that the integrand is $\bigO(w^{-2}\ln w )$ as $w \to \infty$. Thus we have
\begin{align*}
X_{2,r}&=  \int_{\gamma_\infty} \ln \Gamma \bigg( \frac{1+\nu}{2} -w \bigg)  \bigg\{ \psi\bigg( \frac{1+\nu}{2}+w\bigg)-\ln(w) -f(w)\bigg\} \frac{dw}{2\pi i} + \bigO \bigg( \frac{\ln (s^\rho) }{s^\rho} \bigg)
\\
&=: X_{2,r}^{(c)} +  \bigO \bigg( \frac{\ln (s^\rho) }{s^\rho} \bigg)
\end{align*}
as $s \to +\infty$, where $f(w)$ is defined by \eqref{def of f} and where the contour $\gamma_\infty$ is a line oriented upwards and approaching infinity which crosses the real line between the origin and $m=\frac{1+\nu}{2}$ (see Figure \ref{fig: gamma infty}). 
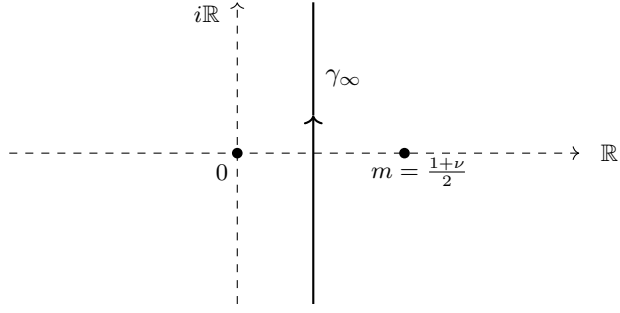
\begin{figure}
	\begin{center}
		\begin{tikzpicture}	
		
		\draw[dashed,->] (5,-2) -- (5,2);
		\draw[dashed,->] (2,0) -- (9.5,0);
		\node at (4.6,1.85) {\small $i \mathbb{R}$};
		\node at (9.9,0) {\small $\mathbb{R}$};
		
		\draw[thick,->] (6,-2) -- (6,0.5);
		\draw[thick] (6,0.5) -- (6,2);
		\node at (6.4,1) {$\gamma_\infty$};

		\fill (5,0) circle (0.07cm);
		\node at (4.8,-0.25) {\small $0$};
		
		\fill (7.2,0) circle (0.07cm);
		\node at (7.4,-0.25) {\small $m=\frac{1+\nu}{2}$};
		\end{tikzpicture}
		\caption{The contour $\gamma_\infty$.}  \label{fig: gamma infty}
	\end{center}
\end{figure}
We will compute $X_{2,r}$ by integration and then contour deformation. However, we first need to add and subtract the term $\frac{\nu}{2w}$ in the integrand and to split $X_{2,r}$ into two parts as follows:
\begin{align*}
X_{2,r}^{(c)}&=  \int_{\gamma_\infty} \ln \Gamma \bigg( \frac{1+\nu}{2} -w \bigg)  \bigg\{ \psi\bigg( \frac{1+\nu}{2}+w\bigg)-\ln(w) -\frac{\nu}{2w}-\frac{1-3\nu^2}{24(w-m)^2} \bigg\} \frac{dw}{2\pi i} 
\\
& \quad  +\int_{\gamma_\infty} \ln \Gamma \bigg( \frac{1+\nu}{2} -w \bigg)  \bigg\{  \frac{\nu}{2w} - \frac{\nu}{2} \bigg( \frac{1}{w-m} -\frac{m}{(w-m)^2}  \bigg) \bigg\} \frac{dw}{2\pi i}.
\end{align*}
Now, we integrate by parts the first integral, while the second integral can be evaluated explicitly by deforming the contour to infinity on the left half-plane (there is only a residue at $w=0$). This gives
\begin{align*}
X_{2,r}^{(c)}&=  \int_{\gamma_\infty} \psi \bigg( \frac{1+\nu}{2} -w \bigg)  \bigg\{ \ln \Gamma\bigg( \frac{1+\nu}{2}+w\bigg)+w\big( 1-\ln (w) \big) 
\\
&\qquad -\frac{\nu}{2}\ln(w)+\frac{1-3\nu^2}{12(2w-\nu-1)}-\frac{\ln(2\pi)}{2}\bigg\}\frac{dw}{2\pi i} + \frac{\nu}{2} \ln \Gamma \bigg( \frac{1+\nu}{2} \bigg).
\end{align*}
For the remaining integral, we deform the contour to infinity in the right half-plane. Note that this would not have been possible without adding and substracting the term $\frac{\nu}{2w}$. Since $\psi(\frac{1+\nu}{2}-w)$ has simple poles with residue $1$ at $w=m+n$, $n=0,1,2,...$, we pick up the following residue contributions
\begin{align*}
-\bigg\{ \ln \Gamma\big( 1+\nu + n\big)+(n+m)\big( 1-\ln (n+m) \big) -\frac{\nu}{2}\ln(n+m)+\frac{1-3\nu^2}{24n}-\frac{\ln(2\pi)}{2}\bigg\}
\end{align*}
at the points $n+m$ for $n=1,2,\ldots$, and 
\begin{align*}
 -\bigg\{ \ln \Gamma\big( 1+\nu \big)+m\big( 1-\ln (m) \big) -\frac{\nu}{2}\ln(m)-\frac{\ln(2\pi)}{2}\bigg\} +\frac{\gamma_{\mathrm{E}} (1-3\nu^2)}{24}
 \end{align*} 
at $w=m$, where $\gamma_{\mathrm{E}}$ is Euler's gamma constant.
This yields
\begin{align*}
X_{2,r}^{(c)}&= -\sum_{n=0}^{\infty} \bigg\{ \ln \Gamma\big( 1+\nu + n\big)+(n+m)\big( 1-\ln (n+m) \big) -\frac{\nu}{2}\ln(n+m)+\frac{1-3\nu^2}{24(n+1)}-\frac{\ln(2\pi)}{2}\bigg\}
\\
&\qquad +\frac{\gamma_{\mathrm{E}} (1-3\nu^2)}{24}+ \frac{\nu}{2} \ln \Gamma \bigg( \frac{1+\nu}{2} \bigg).
\end{align*}
The series is convergent since it arises from a convergent integral. We rewrite Euler's gamma constant (see \cite[Eq. 5.2.3]{NIST}) as
\begin{align} \label{identity Euler gamma}
\gamma_E=\sum_{n=1}^{\infty}\bigg[ \frac{1}{n} -\ln \bigg(  1+\frac{1}{n}\bigg) \bigg],
\end{align}
which implies
\begin{align}
X_{2,r}^{(c)}= & -\sum_{n=0}^{\infty} \bigg\{ \ln \Gamma\big( 1+\nu + n\big)+(n+m)\big( 1-\ln (n+m) \big)+ \frac{1-3\nu^{2}}{24} \ln\bigg(  1+\frac{1}{n+1}\bigg) \nonumber \\
& -\frac{\nu}{2}\ln(n+m) -\frac{\ln(2\pi)}{2}\bigg\} + \frac{\nu}{2} \ln \Gamma \bigg( \frac{1+\nu}{2} \bigg). \label{first expression X2}
\end{align}
This series can be computed explicitly. From the formula (cf. \cite[Eq. 5.17.1]{NIST})
\begin{align*}
G(z+1) =\Gamma(z) G(z),
\end{align*}
we deduce
\begin{align*}
-\sum_{n=0}^{N} \ln \Gamma\big( 1+\nu + n\big) = -\ln G \big( 2+\nu + N\big) + \ln G \big(1+\nu \big).
\end{align*}
The asymptotic formula \eqref{asymptotics lnGamma and lnG} then implies that
\begin{align}\label{asymptotics for X2 a}
-\sum_{n=0}^{N} \ln \Gamma\big( 1+\nu + n\big) = \Omega (\ln N) - \zeta'(-1)-\frac{\ln (2\pi)}{2}(1+\nu) + \ln G \big(1+\nu \big) + \bigO \big(N^{-1}\big)
\end{align}
as $N \to +\infty$, where $\zeta'(-1)$ denotes the derivative of the Riemann zeta function evaluated at $-1$. Furthermore, from \cite[Eq. (10.11)]{CLMMuttalib} with $\theta=1$, we have
\begin{align}
&\sum_{n=0}^N \bigg( \frac{1+\nu}{2} +n \bigg) \ln  \bigg( \frac{1+\nu}{2} +n \bigg) = \frac{1+\nu}{2} \ln  \bigg( \frac{1+\nu}{2} \bigg)+ \Omega(\ln N)  \nonumber
\\
&\quad + \frac{3(1+\nu^2) +8 +12 \nu}{24} - \zeta '\bigg( -1, \frac{1+\nu}{2}+1  \bigg) + \bigO \big(N^{-1}\big), \qquad N \to +\infty,\label{asymptotics for X2 b}
\end{align}
where we recall that $\zeta(z,u)$ is the Hurwitz zeta function and $\zeta'(-1, m+1)=\partial_{z}\zeta(z, m+1)|_{z=-1}$. Also, it is easy to verify from $\Gamma(z+1)=z\Gamma(z)$ that
\begin{align*}
\sum_{n=0}^N \ln \bigg( \frac{1+\nu}{2}+n \bigg)= \ln \Gamma \bigg( \frac{1+\nu}{2}+N+1 \bigg)- \ln \Gamma \bigg( \frac{1+\nu}{2}\bigg),
\end{align*}
and thus, by \eqref{asymptotics lnGamma and lnG},
\begin{align} \label{asymptotics for X2 c}
\frac{\nu}{2}\sum_{n=0}^N \ln \bigg( \frac{1+\nu}{2}+n \bigg)=  \Omega(\ln N) +\frac{\nu}{4} \ln (2\pi) -\frac{\nu}{2} \ln \Gamma \bigg( \frac{1+\nu}{2}\bigg) + \bigO \big(N^{-1}\big), \qquad N \to +\infty.
\end{align}
Another straightforward calculation shows that
\begin{align} \label{asymptotics for X2 d}
\sum_{n=0}^{N} \ln \bigg( 1+ \frac{1}{n+1} \bigg) = \ln (N+2) = \ln N + \bigO \big(N^{-1}\big), \qquad N \to +\infty.
\end{align}
Substituting \eqref{asymptotics for X2 a}, \eqref{asymptotics for X2 b}, \eqref{asymptotics for X2 c}, and \eqref{asymptotics for X2 d} into \eqref{first expression X2} gives 
\begin{align} \label{asymptotics X2r}
X_{2,r}= X_{2,r}^{(c)} +  \bigO \bigg( \frac{\ln (s^\rho) }{s^\rho} \bigg), \qquad s\to +\infty,
\end{align}
where $X_{2,r}^{(c)} $ is given by \eqref{def of X2rc}.

\paragraph{Asymptotics of $X_{3,r}$.} The integrand of $X_{3,r}$ is analytic on the left of $i s^\rho \Sigma_5$. By deforming the contour $i s^\rho \Sigma_5$ to $is^\rho \gamma_{b_2b_1}$, where $\gamma_{b_2b_1}$ is defined as in \eqref{new expression I_1r}, we rewrite $X_{3,r}$ as follows:
\begin{align*}
X_{3,r}&= -\int_{is^\rho\gamma_{b_2b_1}} \ln \Gamma \bigg( \frac{1+\nu}{2} -w \bigg)  \bigg\{ r \psi\bigg( \frac{1+\nu}{2}-w \bigg)-r\ln(-w)+f(w)\bigg\} \frac{dw}{2\pi i}
\\
&=  -\int_{is^\rho\gamma_{b_2b_1}} \ln \Gamma \bigg( \frac{1+\nu}{2} -w \bigg)  \bigg\{ r \psi\bigg( \frac{1+\nu}{2}-w \bigg)-r\ln(-w)+ \frac{\nu}{2w} + \frac{1-3 \nu^2}{24w^2}\bigg\} \frac{dw}{2\pi i} + \bigO \bigg( \frac{\ln (s^\rho)}{s^\rho}\bigg)
\end{align*}
as $s \to +\infty$. Using the asymptotic formulas \eqref{asymptotics lnGamma and lnG} and \eqref{asymtptoics for psi} for the integrand of $X_{3,r}$, we find
\begin{align*}
X_{3,r}= & -\int_{is^\rho\gamma_{b_2b_1}} \bigg\{ \frac{\nu}{2} (r-1)\big(\ln(-w)-1\big) 
\\
& + \frac{1+r-3\nu^2(1+r)-6(r-1)\nu \ln(2\pi)-(1+r+3(r-3)\nu^2)\ln(-w)}{24w} 
\\
&+\bigO \big(w^{-2}\ln(-w)\big)\bigg\} \frac{dw}{2\pi i} + \bigO \bigg( \frac{\ln (s^\rho)}{s^\rho}\bigg), \qquad s\to +\infty.
\end{align*}
After the change of variables $w=is^\rho \zeta$, we obtain
\begin{align}
X_{3,r} &= \Omega \big( \ln (s^\rho) \big)+\frac{1+r+3(r-3)\nu^2}{24} \int_{\gamma_{b_2b_1}} \frac{\ln(-i\zeta)}{\zeta} \frac{d\zeta}{2\pi i} \nonumber
\\
&\quad - \frac{1+r-3\nu^2(1+r)-6(r-1)\nu \ln(2\pi)}{24} \int_{\gamma_{b_2b_1}} \frac{1}{\zeta}\frac{d\zeta}{2\pi i}+ \bigO \bigg( \frac{\ln (s^\rho)}{s^\rho}\bigg) \label{lol6}
\end{align}
as $s \to +\infty$. Since
\begin{align} \label{integral identities X3r}
 \int_{\gamma_{b_2b_1}} \frac{1}{\zeta}\frac{d\zeta}{2\pi i}= \frac{\ln (b_1/b_2)}{2\pi i}, \qquad  \int_{\gamma_{b_2b_1}} \frac{\ln(-i\zeta)}{\zeta} \frac{d\zeta}{2\pi i}= \frac{\ln(b_1/b_2)\ln(-b_1b_2)}{4\pi i},
\end{align}
the asymptotics \eqref{lol6} can be rewritten as
\begin{align} \label{asymptotics of X3r}
X_{3,r} =\Omega \big( \ln (s^\rho) \big)+X_{3,r}^{(c)} + \bigO \bigg( \frac{\ln (s^\rho)}{s^\rho}\bigg), \qquad s \to +\infty,
\end{align}
where $X_{3,r}^{(c)} $ is given by \eqref{def of X3rc}.
\paragraph{Asymptotics of $X_{r}$.} By substituting \eqref{asymptotics of X1r}, \eqref{asymptotics X2r}, and \eqref{asymptotics of X3r} into \eqref{splitting Xr}, we obtain \eqref{asymp for Xr Prop}. This completes the proof of Proposition \ref{prop: asymptotics Xr}.
\end{proof}

\subsubsection{Asymptotics of $Z_r$}
In this subsection, we compute the asymptotics of
\begin{align}\label{def of Zr}
Z_r=-2is^\rho \int_{\gamma_{b_2b_1}}\psi \Big(\frac{1+\nu}{2} -i s^{\rho} \zeta\Big) p(\zeta)\frac{d\zeta}{2\pi i},
\end{align}
where $p(\zeta)$ is defined by \eqref{def of p}. Some of the following computations are similar to those performed in \cite[Section 8]{CLMMuttalib}. In particular, the quantity $\hat{f}_{1}$ in \cite[Eq (3.15)]{CLMMuttalib} is in our case equal to
\begin{align*}
\hat{f}_{1} = -i \frac{3\nu^{2}-1}{24}
\end{align*}
and appears naturally in the asymptotics of $Z_{r}$. 
\begin{proposition}\label{prop: asymptotics of Zr}
Let $\nu > -1$ and $Z_r$ be defined by \eqref{def of Zr}. Then
\begin{align} \label{asymptotics of Zr}
Z_r= \Omega \big(\ln (s^\rho)  \big) + Z_r^{(c)}  +\bigO \bigg( \frac{\ln s^\rho}{s^\rho} \bigg)
\end{align}
as $s \to +\infty$ uniformly for $r$ in compact subsets of $[1,+\infty)$, where
\begin{align} \label{def of Zrc}
Z_{r}^{(c)} &=-2i \Bigg\{\frac{1}{4}  \bigg( \frac{c_8(\arg b_1 - \arg b_2)}{2} + \frac{\pi (c_8-2i\hat{f}_{1})}{2}\bigg(1-\frac{\im b_2}{|b_2|}\bigg) \bigg) \nonumber +\frac{1}{2\pi i} \Bigg( \frac{ic_8}{4}\big( (\ln b_1)^2 - (\ln b_2)^2 \big) \nonumber
\\
& \qquad \quad + \frac{\pi (c_8-2i\hat{f}_{1})}{2 |b_2|} \bigg[ |b_2|-\im b_2 - |b_2| \ln\bigg(\frac{2i |b_2|^2}{|b_2|+ \im b_2}\bigg) +\im b_2 \ln \bigg(\frac{i(|b_2|+\im b_2)}{2} \bigg)  \bigg] \Bigg)\nonumber
\\
&\qquad + \frac{\nu}{16 \pi} \bigg(-4 i \pi \ln(|b_2|)(c_5-c_6)+4i \pi c_6 \ln\bigg(\frac{2}{|b_2|+\im b_2}\bigg) +\ln b_1(-2c_7+i\pi(3c_5+c_6)) \nonumber
\\
& \quad \qquad -(c_5+c_6)\ln^2 (b_1) +\ln(b_2)(2c_7+i \pi (c_5-c_6)) + (c_5+c_6) \ln^2 (b_2) +2 \pi^2 c_5\bigg) \Bigg\}.
\end{align}
\end{proposition}

\begin{proof}
By \eqref{asymtptoics for psi}, we see that
\begin{align} \label{asymptotics of psi with parameter s}
\psi\bigg( \frac{1+\nu}{2} -is^\rho \zeta\bigg)= \ln(s^\rho) + \ln(-i\zeta) -\frac{\nu}{2 i s^\rho \zeta} + \bigO\big(s^{-2\rho}\big),\qquad \mbox{as } s \to +\infty,
\end{align}
uniformly for $\zeta \in \gamma_{b_2 b_1}$. From \cite[beginning of Section 8.3 and Eq. (3.38)]{CLMMuttalib} with the coefficients $c_{1},\ldots,c_{8}$ given by \eqref{eq: values of constants cj Fr}, we have
\begin{align} \label{asymptptics of p}
p(\zeta)= -\frac{c_4}{2 \rho} \ln (s^\rho) +\frac{\mathcal{B}(\zeta)}{2} + \frac{\mathcal{A}(\zeta)}{s^\rho} + \bigO\big(s^{-2\rho}\big)
\end{align}
as $s\to +\infty$ uniformly for $\zeta \in \gamma_{b_2b_1}$ and $r$ in compact subsets of $[1,+\infty)$, where
\begin{align} \label{def of mathcalB}
\mathcal{B}(\zeta)&= -c_{5}r(\zeta) \int_{0}^{i\infty} \frac{d\xi}{r(\xi)(\xi-\zeta)} - c_{6}r(\zeta) \int_{0}^{-i\infty} \frac{d\xi}{r(\xi)(\xi-\zeta)} - c_{5} \ln(i\zeta)-c_{6}\ln(-i\zeta)-c_{7},
\\
\mathcal{A}(\zeta) &= \frac{i c_8}{2\zeta} +r(\zeta) \frac{c_8-\frac{3\nu^2 -1}{12} }{2 \zeta |b_2|}.
\end{align}
After substituting \eqref{asymptotics of psi with parameter s} and \eqref{asymptptics of p} into the definition \eqref{def of Zr} of $Z_r$, we write
\begin{align} 
Z_r & = \Omega\big( \ln (s^\rho) \big)  -2i \bigg( \int_{\gamma_{b_2b_1}} \mathcal{A}(\zeta)\big(\ln(\zeta) - \tfrac{\pi i}{2}\big) \frac{d\zeta}{2\pi i} - \frac{\nu}{4i}\int_{\gamma_{b_2b_1}}\frac{\mathcal{B}(\zeta)}{\zeta} \frac{d\zeta}{2\pi i} \bigg) + \bigO \bigg( \frac{\ln s^\rho}{s^\rho} \bigg) \label{asymptotics Zr intermediate}
\end{align}
as $s \to +\infty$. It was shown in \cite[Lemma 8.4]{CLMMuttalib} that
\begin{align}
 \int_{\gamma_{b_2b_1}} \mathcal{A}(\zeta) d\zeta  \label{calAintegral}
= &\;  -\frac{c_8(\arg b_1-\arg b_2)}{2} - \frac{\pi \big( c_8-\frac{3\nu^2-1}{12} \big)}{2} \bigg( 1-\frac{\im b_2}{|b_2|} \bigg),
\\
 \int_{\gamma_{b_2b_1}} \ln (\zeta) \mathcal{A}(\zeta) d\zeta  \nonumber
= &\; \frac{ic_8}{4}((\ln{b_1})^2 - (\ln{b_2})^2) 
+ \frac{\pi \big(c_8-\frac{3\nu^2-1}{12}\big)}{2|b_2|}\bigg\{|b_2| - \im{b_2} - |b_2| \ln\bigg(\frac{2i|b_2|^2}{|b_2| + \im{b_2}}\bigg)
\\ \label{lnzetacalAintegral}
& + (\im{b_2})\ln\bigg(\frac{i(|b_2| + \im{b_2})}{2}\bigg)\bigg\},
\\
\int_{\gamma_{b_2b_1}} \frac{\mathcal{B}(\zeta)}{\zeta} d\zeta 
= &\; \frac{1}{2} \bigg\{-4 i \pi  \ln(|b_2|) (c_5-c_6)+4 i \pi  c_6 \ln \left(\frac{2}{|b_2|+\im(b_2)}\right) \nonumber
\\\nonumber
& +\ln(b_1) (-2 c_7+i \pi  (3 c_5+c_6))-(c_5+c_6)(\ln{b_1})^2 
\\\label{calBoverzetaintegral}
& +\ln(b_2) (2 c_7+i \pi  (c_5-c_6))+ (c_5+c_6)(\ln{b_2})^2 +2 \pi ^2 c_5\bigg\}.
\end{align}
By substituting \eqref{calAintegral}, \eqref{lnzetacalAintegral}, and \eqref{calBoverzetaintegral} into \eqref{asymptotics Zr intermediate}, we obtain \eqref{asymptotics of Zr} after a long computation. This completes the proof of the proposition.
\end{proof}

\subsection{Asymptotics of $I_{3,r}(K)$ and $I_{4,r}(K)$}
In this section we compute the large $s$ asymptotics of $I_{3,r}(K)$ and $I_{4,r}(K)$ defined in \eqref{def of I3r} and \eqref{def of I4r}.
\begin{proposition} \label{prop: asymptotics I3r}
Let $\nu > -1$ and $K=s^\rho$. Then
\begin{align}\label{asymp I3r in prop}
I_{3,r}(K) = \Omega \big( \ln(s^\rho) \big) + I_{3,r}^{(c)} + \bigO \bigg( \frac{\ln (s^\rho)}{s^\rho} \bigg)
\end{align}
as $s \to +\infty$ uniformly for $r$ in compact subsets of $[1,+\infty)$, where
\begin{align} \nonumber
I_{3,r}^{(c)}=\frac{-2(1+r)+3r(r^2-1)\nu^2 +2r(1+3r\nu^2)\ln(r)}{24r (1+r)^2}.
\end{align}
\end{proposition}

\begin{proof}
Proceeding as in the proof of \cite[Proposition 9.1]{CLMMuttalib}, one obtains
\begin{align*}
I_{3,r}(K)= -\frac{1}{2 s^\rho} \int_{\tilde{\sigma}} \ln \Gamma \bigg( \frac{1+\nu}{2} - is^\rho \zeta \bigg) W(\zeta) \frac{d\zeta}{2\pi i} + \bigO \bigg(  \frac{\ln (s^\rho)}{s^\rho}\bigg), \qquad s\to + \infty,
\end{align*}
where $\tilde{\sigma}$ surrounds the horizontal segment $[b_1,b_2]$ but does not surround the origin, and is oriented counterclockwise. The function $W(\zeta)$ is defined by
\begin{align} \nonumber
W(\zeta) &= \frac{1}{\tilde{r}(\zeta)} \tr \bigg[ \bigg(- \frac{A}{(\zeta-b_1)^2} - \frac{2B}{(\zeta-b_1)^3} + \frac{\bar A}{(\zeta-b_2)^2} - \frac{2\bar B}{(\zeta-b_2)^3} \bigg) 
\\ \label{def of W}
& \qquad \times \begin{pmatrix}
\zeta - i \im b_2 &  i \re b_2 \\ i \re b_2 & i \im b_2 -\zeta 
\end{pmatrix} \bigg],
\end{align}
where $\tilde{r}(\zeta)= \sqrt{(\zeta-b_1)(\zeta-b_2)}$ has a branch cut on $[b_1,b_2]$, such that $\tilde r(\zeta) \sim \zeta$ as $\zeta \to \infty$. The matrices $A$ and $B$ denote the coefficients appearing in the large $s$ asymptotics of $R$ (cf. \cite[Proposition 4.1]{CLMMuttalib}) and are given by
\begin{align}\label{ABexplicit}
A = \begin{pmatrix}
A_{1,1} & A_{1,2} \\ A_{2,1} & A_{2,2}
\end{pmatrix},
\qquad 
B =-\frac{5b_1}{48(c_1+c_2)} \begin{pmatrix}
i & 1 \\ 1 & -i
\end{pmatrix},
\end{align}
with
\begin{align*}
A_{1,1}&=\frac{3 \im b_2 +2i\re b_2-12(|b_2|(c_5-c_6)(c_5+c_6)+(c_5^2+c_6^2)\im b_2+2ic_5c_6 \re b_2)}{48 (c_1+c_2) \re b_2}, 
\\
A_{1,2}&=\frac{4i(3|b_2|(c_5-c_6)(1+c_5+c_6)+\im b_2 + 3(c_5+c_5^2+c_6+c_6^2)\im b_2)}{48(c_1+c_2)\re b_2}
\\
&\quad -\frac{(5+12c_6+12c_5(1+2c_6))}{48(c_1+c_2)},
\\
A_{2,1}&= \frac{12i|b_2|(c_5-c_6)(-1+c_5+c_6)+4i(1+3(c_5-1)c_5+3(c_6-1)c_6) \im b_2}{48(c_1+c_2)\re b_2}
\\
& \quad +\frac{-5+12(c_5+c_6-2c_5c_6)}{48(c_1+c_2)},
\\
A_{2,2}&=-A_{1,1}.
\end{align*}
Using the asymptotics \eqref{asymptotics lnGamma and lnG} of $\ln \Gamma$, we obtain
\begin{align}\label{lol7}
I_{3,r}(K)= -\frac{i}{2} \int_{\tilde{\sigma}} \big(1-\ln(-is^\rho \zeta)\big) \zeta W(\zeta) \frac{d\zeta}{2\pi i} + \bigO \bigg(  \frac{\ln (s^\rho)}{s^\rho}\bigg), \qquad s\to + \infty.
\end{align}
From \eqref{def of W}, we see that
\begin{align*}
W(\zeta)= - \frac{\tr [(A-\bar{A})\sigma_3]}{\zeta^2} + \bigO \big( \zeta^3 \big) =-i\frac{1-12c_5c_6}{6(c_1+c_2)\zeta^2}+ \bigO \big( \zeta^3 \big) \qquad \mbox{as } \zeta \to \infty.
\end{align*}
Thus, after splitting the leading term in \eqref{lol7} into two parts, we obtain
\begin{align} \label{intermediate I3rK}
I_{3,r}(K)=-i\frac{1-12c_5c_6}{6(c_1+c_2)}\frac{-i}{2}\bigg(1-\ln (s^\rho) \bigg) + \frac{i}{2} \int_{\tilde{\sigma}} \ln(-i\zeta) \zeta W(\zeta) \frac{d\zeta}{2\pi i} + \bigO \bigg(  \frac{\ln (s^\rho)}{s^\rho}\bigg)
\end{align}
as $s\to + \infty$, where we have deformed $\tilde{\sigma}$ to infinity for the first part. The last integral in \eqref{intermediate I3rK} can be evaluated as follows:
\begin{align*}
\frac{i}{2} \int_{\tilde{\sigma}} \ln(-i\zeta) \zeta W(\zeta) \frac{d\zeta}{2\pi i}
&= \lim_{R\to \infty} \bigg\{ \frac{i}{2}\int_{C_R}\ln(-i\zeta)\zeta W(\zeta)\frac{d\zeta}{2 \pi i} +  \frac{i}{2}\int_{-iR}^0 \zeta W(\zeta) d\zeta   \bigg\}
\\
&=\lim_{R\to \infty} \bigg\{ -\frac{i}{2}\tr [(A-\bar{A})\sigma_3] \int_{C_R}\frac{\ln(-i\zeta)}{\zeta}\frac{d\zeta}{2 \pi i} +  \frac{i}{2}\int_{-R}^0 it W(it) idt  \bigg\},
\end{align*}
where $C_{R}$ is the circle centered at the origin of radius $R$ oriented positively. Since
\begin{align*}
& \int_{C_R}\frac{\ln(-i\zeta)}{\zeta}\frac{d\zeta}{2 \pi i} = \ln(R), \\
& \tilde{r}(it)= -i \sqrt{(t-\im b_2)^2 + (\re b_2)^2}, \qquad \mbox{for }t<0,
\end{align*}
we compute the integral $\int_{-iR}^0 \zeta W(\zeta) d\zeta$ by a rather long primitive calculation, which uses the definition \eqref{def of W} of $W(\zeta)$. Then, after substituting the expressions \eqref{eq: values of constants cj} and \eqref{b1b2def}, we obtain
\begin{align}\label{lol8}
\frac{i}{2} \int_{\tilde{\sigma}} \ln(-i \zeta) \zeta W(\zeta) \frac{d\zeta}{2\pi i} &= \frac{(1+r)\big(-2+r(2+(9r-3)\nu^2)\big)+2r(1+3r\nu^2)\ln (r)}{24 r(1+r)^2}.
\end{align}
Substituting \eqref{lol8} into \eqref{intermediate I3rK} and using again \eqref{eq: values of constants cj Fr}, we obtain \eqref{asymp I3r in prop}, which finishes the proof.
\end{proof}

\begin{proposition}\label{prop: asymptotics I4r}
Let $\nu > -1$ and $K=s^\rho$. Then, for any integer $N\ge 1$,
\begin{align}
I_{4,r}(K) = \bigO \big( s^{-N\rho} \big)
\end{align}
as $s \to +\infty$ uniformly for $r$ in compact subsets of $[1,+\infty)$.
\end{proposition}
\begin{proof}
The proof is analogous to the proof of \cite[Proposition 9.2]{CLMMuttalib} and relies on the large $s$ asymptotics of $R$.
\end{proof}

\subsection{Integration of the differential identity in $r$}
\label{subsec: integration1}
In this subsection, we compute the constant term $C_{r}$ in the large gap asymptotics for the point process defined by $\mathbb{K}_r$. From \cite{ClaeysGirSti,CLMMuttalib}, these asymptotics are of the form
\begin{align} \label{asymptotics det Kr}
\det \big( 1- \mathbb{K}_r|_{[0,s]}\big) = C_r \exp \bigg( -a_r s^{2 \rho} + b_rs^{\rho} +c_r \ln s + \bigO \big( s^{-\rho} \big) \bigg),
\end{align}
where the constants $\rho$, $a_r$, $b_r$ and $c_r$ are given by
\begin{align*} 
& \rho=\frac{1}{1+r}, \qquad a_{r}=\frac{r^{\frac{1-r}{1+r}}(r+1)^2}{4}, \qquad  b_{r}=(1+r)r^{\frac{1}{1+r}} \nu, \qquad c_{r}=\frac{r-1}{12(r+1)}-\frac{r \nu^2}{2(r+1)}.
\end{align*} 
\begin{proposition} \label{prop: asymptotics of Kr}
Let $\nu > -1$ and $r\geq 1$. Then
\begin{align*}
\det \big( 1- \mathbb{K}_r|_{[0,s]}\big) = C_r \exp \bigg( -a_r s^{2 \rho} + b_rs^{\rho} +c_r \ln s + \bigO \big( s^{-\rho} \big) \bigg),
\end{align*}
where
\begin{align*}
C_r&=\frac{G(1+\nu)^r}{ (2\pi)^{\frac{r \nu}{2}}} \exp\big\{ -(r-1)\zeta'(-1)  \big\}
\exp \bigg\{ \frac{-2+r^2(r-1+12\nu^2)}{24(r+1)} \ln(r) \bigg\}
\\
&\quad \times \exp\bigg\{  - \frac{(r-1)^2+12r\nu^2}{24}\ln(1+r) \bigg\}.
\end{align*}
\end{proposition}
\begin{proof}
It follows from the analysis of \cite{ClaeysGirSti,CLMMuttalib} that the error term in \eqref{asymptotics det Kr} can be differentiated with respect to $r$ and that its $r$-derivative is of order $\bigO(s^{-\rho}\ln (s^\rho))$ uniformly for $r$ in compact subsets of $[1,+\infty)$. Therefore we have
\begin{align*}
\partial_r \ln \det \big( 1- \mathbb{K}_r|_{[0,s]}\big) &= -2 \partial_{r}(\rho) a_r s^{2\rho} \ln s - \partial_{r} a_r s^{2\rho} + \partial_{r}(\rho) b_r s^{\rho} \ln s + \partial_{r}(b_r) s^{\rho} \\
&\qquad  + \partial_{r}(c_r) \ln(s) + \partial_{r}(\ln C_r) + \bigO(s^{-\rho}\ln (s^\rho))
\end{align*}
as $s \to + \infty$, uniformly for $r$ in compact subsets of $[1,+\infty)$. 
Thus, from Lemma \ref{lemma: differential identities} and Propositions \ref{prop: I_1r}, \ref{prop: asymptotics Xr}, \ref{prop: asymptotics of Zr}, \ref{prop: asymptotics I3r}, and \ref{prop: asymptotics I4r}, we infer that
\begin{align*}
\partial_{r}(\ln C_r)&= I_{1,r}^{(c)} + X_{1,r}^{(c)} + X_{2,r}^{(c)} +X_{3,r}^{(c)}+Z_r^{(c)} + I_{3,r}^{(c)}
\\
&=- \frac{2-2r+3r \nu^2}{24r}-\frac{\nu \ln(2\pi)}{4} +\frac{1+r(r-1+r^2+6(2+r)\nu^2)}{12(r+1)^2}\ln(r)
\\
&\quad -\frac{1}{12}(r-1+6\nu^2)\ln(1+r) + \frac{\nu-1}{2} \ln \Gamma \bigg(\frac{1+\nu}{2}  \bigg)  -\ln G\bigg( \frac{1+\nu}{2}\bigg)-\frac{\nu\ln(2\pi)}{4}
\\
&\quad +\ln G(1+\nu)+ \frac{3\nu^2-1}{24}-\zeta'\bigg( -1; \frac{1+\nu}{2}+1\bigg) +\frac{\nu+1}{2}\ln \bigg( \frac{\nu+1}{2}\bigg).
\end{align*}
Applying the identity \cite[Eq (18)]{A1998}
\begin{align} \label{identity Hurwitz Barnes Gamma}
\ln G(z+1) = \zeta'(-1)-\zeta'(-1,z+1)+z \ln \Gamma(z+1)
\end{align}
with $z = \frac{1+\nu}{2}$, we get
\begin{align*}
\partial_{r}(\ln C_r)&=\frac{r-2}{24r}-\frac{\nu \ln(2\pi)}{2} +\frac{1+r(r-1+r^2+6(2+r)\nu^2)}{12(r+1)^2}\ln(r)
\\
&\quad -\frac{1}{12}(r-1+6\nu^2)\ln(1+r)- \zeta'(-1)+\ln G(1+\nu).
\end{align*}
Integrating this identity with respect to $r$ from $r=1$ to a fixed $r \geq 1$, we obtain
\begin{align*}
\int_{1}^{r} \partial_{r'}(\ln C_{r'})dr' &= \frac{\nu^2}{2}\ln(2)- \frac{\nu(r-1) \ln(2\pi)}{2}+ (r-1)\bigg( \ln G(1+\nu)- \zeta'(-1) \bigg)
\\
&\quad +\frac{-2+r^2(r-1+12\nu^2)}{24(r+1)} \ln(r) - \frac{(r-1)^2+12r\nu^2}{24}\ln(1+r).
\end{align*}
Since $\ln C_1=\ln G(1+\nu)- \frac{\nu}{2}\ln(2\pi)- \frac{\nu^{2}}{2}\ln 2$ by \eqref{Asymptotics Bessel}, we arrive at
\begin{align*}
\ln C_r&= r\ln G(1+\nu) - r\frac{\nu \ln(2\pi)}{2}- (r-1) \zeta'(-1)
\\
&\quad +\frac{-2+r^2(r-1+12\nu^2)}{24(r+1)} \ln(r) - \frac{(r-1)^2+12r\nu^2}{24}\ln(1+r)
\end{align*}
and the proposition follows by exponentiating both sides.
\end{proof}

\section{Asymptotics of the differential identities in $\nu_\ell$ and $\mu_\ell$}  \label{Section: Asymptotics nu mu}
In this section, we compute the large $s$ asymptotics of the differential identities \eqref{diff identity munu} and \eqref{diff identity munu2}. For $\alpha \in \{\nu_{1},\ldots,\nu_{r},\mu_{1},\ldots,\mu_{q}\}$, these identities express $\partial_\alpha  \ln \det  ( 1- \mathbb{K}|_{[0,s]})$ in terms of the quantities $I_{1,\alpha}$, $I_{2,\alpha}$, $I_{3,\alpha}(K)$, and $I_{4,\alpha}(K)$ defined in Lemma \ref{lemma: differential identities}. By computing the asymptotics of these quantities and then integrating with respect to $\alpha$, we can deduce the large $s$ asymptotics of
\begin{align*}
  \ln \det \big( 1- \mathbb{K}|_{[0,s]}\big),
\end{align*}
see also the outline in Section \ref{outlinesubsec}; in particular \eqref{second diff identity outline}-\eqref{third diff identity outline}. 

In the remainder of the paper, we let $r>q\geq 0$ be integers and let $\nu_1,\ldots, \nu_r,\mu_1, \ldots, \mu_q > -1$ be the parameters associated to the kernel $\mathbb{K}$ defined in (\ref{Meijer G kernel}). We set $\nu_{\min} := \min \{ \nu_1,\ldots, \nu_r,\mu_{1},\ldots,\mu_{q} \}$. The constants $c_1,\ldots,c_8$ and $b_1,b_2$ are defined in \eqref{eq: values of constants cj} and \eqref{b1b2def}, respectively, and we will use the notation $\Omega$ introduced at the beginning of Section \ref{Section: Asymptotics r}.



\subsection{Asymptotics of $I_{1,\alpha}$}
\begin{proposition} \label{prop: I_1alpha}
	Let $I_{1,\alpha}$ be the function defined in \eqref{I1alphadef}. Then
	\begin{align} \label{asymtptotics I_1alpha}
	I_{1,\alpha}= \Omega\big(\ln (s^\rho)\big)+ I_{1,\alpha}^{(c)} + \mathcal{O}\bigg( \frac{\ln ( s^\rho )}{s^\rho} \bigg)
	\end{align}
	as $s\to +\infty$ uniformly for $\alpha$ in compact subsets of $(-1,+\infty)$, where
	\begin{align} \label{I_1alpha3 explicit} \nonumber
	I_{1,\alpha}^{(c)}&=-\frac{(c_1+c_2)}{2}\bigg( \frac{\ln(2\pi)}{2}-\ln\Gamma\Big(\frac{1+2\alpha-\nu_{\min}}{2} \Big)\bigg)\bigg(1-\frac{\im b_2}{|b_2|} \bigg)
	\\
	&\quad -(c_1+c_2)\frac{2\alpha- \nu_{\min}}{2 } \Bigg(\frac{\ln(|b_2|+\im b_2)-\ln(2)}{2}-\frac{\im b_2}{2 |b_2|} \ln\bigg( \frac{2 |b_2|^2}{|b_2|+\im b_2}\bigg) \Bigg).
	\end{align}
\end{proposition}
\begin{proof}
Integrating $I_{1,\alpha}$ by parts as in the proof of Proposition \ref{prop: I_1r}, we find
\begin{align} \label{new expression I1alpha}
I_{1,\alpha} &=  -i(c_1+c_2)\int_{\gamma_{b_2b_1}} \tilde \Psi(\zeta)\bigg( 1-\frac{i \im b_2}{\zeta } \bigg) \frac{1}{r(\zeta)} \frac{d\zeta}{2\pi i} ,
\end{align}
where $\gamma_{b_2b_1}$ is defined as in \eqref{new expression I_1r} and
\begin{align*}
\tilde \Psi(\zeta)= & -s^\rho\int_0^\zeta \psi \Big(\frac{1+2\alpha-\nu_{\min}}{2} - i s^{\rho} \xi \Big)d\xi 
	\\
= & -i\ln \Gamma \Big(\frac{1+2\alpha-\nu_{\min}}{2} - i s^{\rho} \zeta \Big)+i\ln\Gamma\Big(\frac{1+2\alpha-\nu_{\min}}{2} \Big).
\end{align*}
A direct computation using the asymptotics \eqref{asymptotics lnGamma and lnG} yields
\begin{align}\label{lol9}
\tilde \Psi(\zeta) =\Omega\big(\ln (s^\rho)\big)  - i\frac{2\alpha- \nu_{\min}}{2 } \ln(-i \zeta) -i \frac{\ln(2\pi)}{2} +i\ln\Gamma\Big(\frac{1+2\alpha-\nu_{\min}}{2} \Big) + \mathcal{O}\bigg( \frac{\ln ( s^\rho )}{s^\rho} \bigg)
\end{align}
as $s \to + \infty$ uniformly for $\alpha$ in compact subsets of $(-1,+\infty)$, and uniformly for $\zeta \in \gamma_{b_{2}b_{1}}$. We obtain \eqref{asymtptotics I_1alpha}-\eqref{I_1alpha3 explicit} after substituting \eqref{lol9} into \eqref{new expression I1alpha}, using \eqref{integral identities I1r} and then simplifying.
\end{proof}

\subsection{Asymptotics of $I_{2,\alpha}$}
Let $I_{2,\alpha}$ be the function defined in \eqref{I2alphadef}. Using \eqref{lol10}, we can write
\begin{align*}
I_{2,\alpha}=\int_\sigma \mathcal{F}(\zeta)p'(\zeta)\frac{d{\zeta}}{2\pi i},
\end{align*}
where $p(\zeta)$ is defined by \eqref{def of p} and $\mathcal{F}(\zeta)$ is defined by
\begin{align}\label{calFdef}
\mathcal{F}(\zeta)=-\psi \bigg( \frac{1+2\alpha-\nu_{\min} }{2} -is^\rho \zeta \bigg).
\end{align}
We integrate by parts and then deform the contour and use the jumps for $p$ as in the beginning of Section \ref{subsection: asymp for I2r}. This yields
\begin{align*}
I_{2,\alpha}= Z_\alpha + X_\alpha,
\end{align*}
where
\begin{align} \label{def of Xalpha and Zalpha}
Z_\alpha=-2\int_{\gamma_{b_2b_1}} \mathcal{F}'(\zeta)p(\zeta) \frac{d\zeta}{2\pi i}, \qquad X_\alpha =\int_{\Sigma_5} \mathcal{F}'(\zeta)\ln \mathcal{G}(\zeta)\frac{d\zeta}{2\pi i}.
\end{align}
\begin{proposition}\label{prop: asymptotics Xalpha}
	Let $X_\alpha$ be defined by \eqref{def of Xalpha and Zalpha}. Then
\begin{align}
X_{\alpha} &=\Omega \big( \ln (s^\rho) \big) +X_{1,\alpha}^{(c)}+X_{2,\alpha}^{(c)}+X_{3,\alpha}^{(c)}+ \mathcal{O}\bigg( \frac{\ln ( s^\rho )}{s^\rho} \bigg) \qquad \mbox{as } s \to + \infty,
\end{align}
uniformly for $\alpha$ in compact subsets of $(-1,+\infty)$, where
	\begin{align}
	X_{1,\alpha}^{(c)}&=\frac{1}{2\pi i} \bigg\{c_{5} \big( \ln(-ib_{1})\ln(ib_{1})-\ln(-ib_{2})\ln(ib_{2}) \big) \nonumber \\
	& + c_{6} \big( \ln^{2}(-ib_{1})-\ln^{2}(-ib_{2}) \big) + c_{7} \big( \ln(-ib_{1})-\ln(-ib_{2}) \big) \bigg\}, \label{def of X1ac}
	\\
	X_{2,\alpha}^{(c)}&= \sum_{k=0}^{\infty} \bigg\{ \ln\bigg( k+ \frac{1+2\alpha-\nu_{\min} }{2} \bigg) -\psi\big( k+1+\alpha\big) +\frac{\nu_{\min}}{2}\ln\bigg(1+\frac{1}{k+1} \bigg) \bigg\}, \label{def of X2ac} 	\\
	X_{3,\alpha}^{(c)}&= \bigg\{- \frac{\nu_{\min}}{2} +\sum_{j=1}^{r} \frac{2\nu_j -\nu_{\min}}{2} -\sum_{k=1}^{q} \frac{2\mu_k - \nu_{\min}}{2}\bigg\} \frac{\ln(b_1/b_2) \ln(|b_2|)}{2\pi i}. \label{def of X3ac}
	\end{align}
\end{proposition}

\begin{proof}
Recalling the definition \eqref{def of G} of $\mathcal{G}(\zeta)$ and using the identity $\rho^{-1} = c_{1}+c_{2}$, we see that
\begin{align*}
\ln \mathcal{G}(\zeta) &= \ln \Gamma \bigg( \frac{1+\nu_{\min}}{2}+is^\rho \zeta \bigg) + \sum_{k=1}^{q} \ln \Gamma \bigg(  \frac{1+2\mu_k-\nu_{\min}}{2}-is^\rho \zeta\bigg) -\sum_{j=1}^{r} \ln \Gamma \bigg(   \frac{1+2\nu_j -\nu_{\min}}{2}-is^\rho \zeta \bigg)
\\
&\quad - is^\rho \zeta \big(  c_1  \ln(i s^\rho \zeta) + c_2\ln(-is^\rho\zeta) + c_3 \big).
\end{align*}
Hence, recalling the values \eqref{eq: values of constants cj} of $c_{1}$, $c_{2}$ and $c_{3}$,
\begin{align*}
\frac{\mathcal{G}'(\zeta)}{\mathcal{G}(\zeta)} &= is^\rho \Bigg( \psi\bigg( \frac{1+\nu_{\min}}{2}+is^\rho \zeta \bigg) -\sum_{k=1}^{q} \psi \bigg(  \frac{1+2\mu_k-\nu_{\min}}{2}-is^\rho \zeta\bigg)  +\sum_{j=1}^{r} \psi \bigg(   \frac{1+2\nu_j -\nu_{\min}}{2}-is^\rho \zeta \bigg)
\\
&\quad - \ln(is^\rho \zeta) -(r-q) \ln(-is^\rho \zeta)\Bigg).
\end{align*}
Therefore, after integrating by parts, we can write
\begin{align*}
X_\alpha = X_{1,\alpha}+X_{2,\alpha}+X_{3,\alpha},
\end{align*}
where
\begin{align*}
X_{1,\alpha}&=\frac{\mathcal{F}(\zeta)\ln \mathcal{G}(\zeta)}{2\pi i} \Bigg|_{\zeta=b_1}^{b_2},
\\
X_{2,\alpha}&= \int_{is^\rho\Sigma_5} \psi \bigg( \frac{1+2\alpha-\nu_{\min} }{2} -w \bigg) \bigg\{   \psi\bigg( \frac{1+\nu_{\min}}{2}+w \bigg) -\ln(w) - \frac{\nu_{\min}}{2 (w-m)} \bigg\} \frac{dw}{2\pi i},
\\
X_{3,\alpha}&=  \int_{is^\rho \Sigma_5} \psi \bigg( \frac{1+2\alpha-\nu_{\min} }{2} -w \bigg) \bigg\{   -\sum_{k=1}^{q} \psi \bigg(  \frac{1+2\mu_k-\nu_{\min}}{2}-w\bigg)  
\\
&\quad \qquad +\sum_{j=1}^{r} \psi \bigg(   \frac{1+2\nu_j -\nu_{\min}}{2}-w\bigg) -(r-q) \ln(-w) + \frac{\nu_{\min}}{2 (w-m)}\bigg\} \frac{dw}{2\pi i}.
\end{align*}
Here we have used the change of variables $is^{\rho}\zeta = w$ in the expressions for $X_{2,\alpha}$ and $X_{3,\alpha}$, and $m$ is an arbitrary constant which lies in $\mathbb{C}\setminus i s^{\rho}\Sigma_{5}$; it will be convenient to henceforth choose $m=\frac{1+2\alpha - \nu_{\min}}{2}$. As in the proof of Proposition \ref{prop: asymptotics Xr}, we have added and substracted one term in order to ensure that $X_{2,\alpha}$ has a limit as $s \to + \infty$.
It remains to compute the large $s$ asymptotics of  $X_{1,\alpha}$, $X_{2,\alpha}$, and $X_{3,\alpha}$.

\paragraph{Asymptotics of $X_{1,\alpha}$.} By combining \eqref{asymptotics G} and \eqref{asymptotics of psi with parameter s}, we obtain directly that
\begin{align*}
X_{1,\alpha}= \Omega(\ln(s^\rho)) + X_{1,\alpha}^{(c)} +\bigO\bigg( \frac{\ln ( s^\rho )}{s^\rho} \bigg), \qquad s \to +\infty,
\end{align*}
where $ X_{1,\alpha}^{(c)} $ is given by \eqref{def of X1ac}.

\paragraph{Asymptotics of $X_{2,\alpha}$.} The expansion \eqref{asymptotics of psi with w} implies that
\begin{align*}
X_{2,\alpha} = \int_{-i\infty}^{i\infty} \psi \bigg( \frac{1+2\alpha-\nu_{\min} }{2} -w \bigg) \bigg\{   \psi\bigg( \frac{1+\nu_{\min}}{2}+w \bigg) -\ln(w) - \frac{\nu_{\min}}{ 2(w-m)} \bigg\} \frac{dw}{2\pi i} + \bigO\bigg( \frac{\ln ( s^\rho )}{s^\rho} \bigg)
\end{align*}
as $s\to +\infty$. We deform the contour of integration to infinity in the right half-plane and pick up infinitely many residue contributions at the points $m+k$, $k =1,2,\ldots$, of the form
\begin{align*}
- \bigg\{   \psi\big( k+1+\alpha\big) -\ln\bigg( k+ \frac{1+2\alpha-\nu_{\min} }{2} \bigg) - \frac{\nu_{\min}}{ 2k} \bigg\} 
\end{align*}
and one residue contribution $-\psi(1+\alpha)+\ln(m)-\frac{\gamma_{\mathrm{E}} \nu_{\min}}{2}$ at the point $m$. 
Using the identity \eqref{identity Euler gamma}, it follows that
\begin{align*}
X_{2,\alpha} &=-\sum_{k=1}^{\infty}  \bigg\{   \psi\bigg( k+1+\alpha\bigg) -\ln\bigg( k+ \frac{1+2\alpha-\nu_{\min} }{2} \bigg) - \frac{\nu_{\min}}{ 2k} \bigg\} 
\\
&\quad -\psi\big( 1+\alpha\big) +\ln\bigg(  \frac{1+2\alpha-\nu_{\min} }{2} \bigg) -\frac{\gamma_{\mathrm{E}} \nu_{\min}}{2}+ \bigO\bigg( \frac{\ln ( s^\rho )}{s^\rho} \bigg)
\\
&=X_{2,\alpha}^{(c)}+\bigO\bigg( \frac{\ln ( s^\rho )}{s^\rho} \bigg), \qquad s\to +\infty.
\end{align*}

\paragraph{Asymptotics of $X_{3,\alpha}$.} For $X_{3,\alpha}$, we first deform the contour $is^\rho \Sigma_5$ to $is^\rho \gamma_{b_2 b_1}$, then we apply the change of variables $w=is^\rho \zeta$ and the large $s$ asymptotics \eqref{asymptotics of psi with parameter s} of $\psi$. This gives
\begin{align*}
X_{3,\alpha}&=\Omega\big( \ln (s^\rho) \big)+ \bigg\{- \frac{\nu_{\min}}{2} +\sum_{j=1}^{r} \frac{2\nu_j -\nu_{\min}}{2} -\sum_{k=1}^{q} \frac{2\mu_k - \nu_{\min}}{2}\bigg\} \int_{\gamma_{b_2b_1}} \frac{\ln(-i\zeta)}{\zeta}  \frac{d\zeta}{2\pi i}+\mathcal{O} \bigg( \frac{\ln(s^\rho)}{s^\rho} \bigg)
\end{align*}
as $s\to +\infty$. Using the second integral in \eqref{integral identities X3r}, we obtain
\begin{align*}
X_{3,\alpha}&=\Omega\big( \ln (s^\rho) \big)+ X_{3,\alpha}^{(c)}+\mathcal{O} \bigg( \frac{\ln(s^\rho)}{s^\rho} \bigg), \qquad \mbox{as } s\to +\infty.
\end{align*}
\end{proof}

\begin{proposition}\label{prop: asymptotics of Zalpha}
	Let $Z_\alpha$ be defined by \eqref{def of Xalpha and Zalpha}. Then
	\begin{align} \label{asymptotics of Za}
	Z_\alpha= \Omega \big(\ln (s^\rho)  \big) + Z_\alpha^{(c)}  +\bigO \bigg( \frac{\ln s^\rho}{s^\rho} \bigg)
	\end{align}
	as $s \to +\infty$ uniformly for $\alpha$ in compact subsets of $[0,+\infty)$, where
	\begin{align} \label{def of Zac}
	Z_{\alpha}^{(c)} &=\frac{1}{4\pi i} \bigg\{ -4\pi i\ln(|b_2|) (c_5-c_6)+4i\pi c_6 \ln\bigg( \frac{2}{|b_2|+\im b_2} \bigg) \nonumber
	\\
	&\quad +\ln(b_1)(-2c_7 +i \pi (3c_5+c_6))-(c_5+c_6)(\ln b_1)^2 \nonumber
	\\
	&\quad+ \ln(b_2)(2c_7+i \pi (c_5-c_6))+(c_5+c_6)(\ln b_2)^2 +2 \pi^2 c_5\bigg\}.
	\end{align}
\end{proposition}

\begin{proof}
A formula for the large $s$ asymptotics of the function $\mathcal{F}(\zeta)$ defined in (\ref{calFdef}) can be deduced from \eqref{asymptotics of psi with parameter s} (with $\nu$ replaced by $2\alpha-\nu_{\min}$). This formula is uniform for $\zeta \in \gamma_{b_2b_1}$ and can be differentiated with respect to $\zeta$. Hence \begin{align} \label{asymptotics of Fprime}
\mathcal{F}'(\zeta)=-\frac{1}{\zeta} + O(s^{-\rho}) \qquad \mbox{as } \quad s \to +\infty,
\end{align}
uniformly for $\zeta \in \gamma_{b_2b_1}$. Substituting \eqref{asymptotics of Fprime} and the asymptotic formula \eqref{asymptptics of p} for $p(\zeta)$ into the definition \eqref{def of Xalpha and Zalpha} of $Z_\alpha$, we infer that
\begin{align*}
Z_\alpha&=-\frac{c_4}{ \rho} \ln(s^\rho) \int_{\gamma_{b_2b_1}} \frac{1}{\zeta}\frac{d\zeta}{2\pi i} +\int_{\gamma_{b_2b_1}} \frac{\mathcal{B}(\zeta)}{\zeta}\frac{d\zeta}{2\pi i}  +\bigO \bigg( \frac{\ln s^\rho}{s^\rho} \bigg)
\end{align*}
as $s\to +\infty$ uniformly for $\alpha$ in compact subsets of $(-1,+\infty)$, where $\mathcal{B}(\zeta)$ is defined by \eqref{def of mathcalB}. Recalling \eqref{calBoverzetaintegral}, the proposition follows.
\end{proof}

\subsection{Asymptotics of $I_{3,\alpha}(K)$ and $I_{4,\alpha}(K)$}
We next show that the quantities $I_{3,\alpha}(K)$ and $I_{4,\alpha}(K)$ defined in \eqref{I3alphadef} and (\ref{I4alphadef}) vanish as $s$ tends to $+\infty$ for $K = s^{\rho}$.

\begin{proposition}\label{prop: asymptotics I34alpha}
Let $\nu > -1$, $K=s^\rho$, and let $N\ge 1$ be an integer. Then
\begin{align}
I_{3,\alpha}(K) = \bigO \bigg( \frac{\ln (s^\rho)}{s^\rho}\bigg) \quad \text{and} \quad	I_{4,\alpha}(K) = \bigO \big( s^{-N\rho} \big)
\end{align}
as $s \to +\infty$ uniformly for $\alpha$ in compact subsets of $(-1,+\infty)$.
\end{proposition}

\begin{proof}
The proof for $I_{3,\alpha}(K)$ is similar to (but easier than) the proof of Proposition \ref{prop: asymptotics I3r}. In fact, it follows from \eqref{asymptotics of psi with parameter s} that
\begin{align*}
\psi\bigg( \frac{1+2\alpha-\nu_{\min}}{2} -i s^{\rho} \zeta \bigg) = \bigO (\ln (s^\rho))
\end{align*}
as $s \to + \infty$ uniformly for $\zeta \in \tilde{\sigma}$ (see the proof of Proposition \ref{prop: asymptotics I3r} for the definition of $\tilde{\sigma}$), from which we immediately deduce that
\begin{align*}
I_{3,\alpha}(K) = \bigO \bigg( \frac{\ln (s^\rho)}{s^\rho}\bigg), \qquad  \mbox{as } s\to + \infty.
\end{align*}
The proof for $I_{4,\alpha}(K)$ is analogous to the proof of Proposition \ref{prop: asymptotics I4r} and we omit the details.
\end{proof}

\subsection{Integration of the differential identities in $\nu_{\ell}$ and $\mu_{\ell}$} \label{subsec: integration2}
By using the results of the previous subsections, we can compute the large $s$ asymptotics of $\partial_{\alpha} \ln C$ for any $\alpha \in \{\nu_{1},\ldots,\nu_{r},\mu_{1},\ldots,\mu_{q}\}$. Integration with respect to $\alpha$ then yields the following proposition.

%

\begin{proposition} \label{prop: constant after deforming numu}
Let $r>q\geq 0$ be integers and suppose that $\nu_1,\ldots, \nu_r,\mu_1, \ldots, \mu_q >-1$. Let $\nu_{\min} = \min \{ \nu_1,\ldots, \nu_r,\mu_{1},\ldots,\mu_{q} \}$. If $\ell \in \{1,\ldots, r\}$, then
\begin{align} \nonumber
\int_{\nu_{\min}}^{\nu_{\ell}}\partial_{\nu_{\ell}'}(\ln C) d\nu_{\ell}'&=(\nu_{\ell}-\nu_{\min})\bigg(\sum_{\substack{j=1\\j\neq \ell}}^r \nu_j -\sum_{k=1}^{q}\mu_k  \bigg)\big(  \ln(r-q) - \ln(1+r-q)\big)
\\ \nonumber
&\quad +\frac{\nu_{\ell}^2-\nu_{\min}^2}{2}\frac{1+r-q-(r-q)^2 }{1+r-q} \ln(r-q) -(2+q-r)\frac{\nu_{\ell}^2-\nu_{\min}^2}{2} \ln(1+r-q)
\\
&\quad +\ln G(1+\nu_{\ell})-\ln G(1+\nu_{\min}) +(\nu_{\min}-\nu_{\ell} ) \frac{\ln(2\pi)}{2}.\label{after integrating in nu}
\end{align}
If $\ell \in \{1,\ldots, q\}$, then
\begin{align} \nonumber
\int_{\nu_{\min}}^{\mu_{\ell}}\partial_{\mu_{\ell}'}(\ln C) d\mu_{\ell}'&=(\nu_{\min}-\mu_\ell)\bigg(\sum_{j=1}^r \nu_j -\sum_{\substack{k=1\\k\neq \ell}}^{q}\mu_k  \bigg)\big(  \ln(r-q) - \ln(1+r-q)\big)
\\ \nonumber
&\quad +\frac{\mu_\ell^2-\nu_{\min}^2}{2}\frac{1+r-q+(r-q)^2 }{1+r-q} \ln(r-q) -(r-q)\frac{\mu_\ell^2 -\nu_{\min}^2}{2} \ln(1+r-q)
\\
&\quad -\ln G(1+\mu_{\ell})+\ln G(1+\nu_{\min}) -(\nu_{\min}-\mu_\ell)\frac{\ln(2\pi)}{2}.\label{after integrating in mu}
\end{align}
\end{proposition}
\begin{proof}
We only consider the case $\alpha=\nu_{\ell}$, $\ell \in \{1,\ldots,r\}$. The case $\alpha=\mu_{\ell}$, $\ell \in \{1,\ldots,q\}$, is analogous. We start by integrating the term $X_{2,\nu_{\ell}}^{(c)}$ defined in \eqref{def of X2ac}. By Fubini's theorem, we can interchange the order of integration and summation, which implies
\begin{align*}
\int_{\nu_{\min}}^{\nu_{\ell}} X_{2,\nu_{\ell}'}^{(c)} d\nu_{\ell}' =&\; \sum_{k=0}^{\infty}  \Bigg\{   -\ln \Gamma \bigg( k+1+\nu_{\ell}\bigg) -\nu_{\ell}+\bigg( k+ \frac{1+2\nu_\ell-\nu_{\min} }{2} \bigg)\ln\bigg(  k+\frac{1+2\nu_\ell-\nu_{\min} }{2} \bigg)
\\
& +\frac{\nu_{\min}}{2}(\nu_\ell-\nu_{\min})\ln\bigg(1+\frac{1}{k+1} \bigg) +\ln \Gamma \bigg( k+1+\nu_{\min}\bigg) +\nu_{\min}
\\
& -\bigg( k+ \frac{1+\nu_{\min} }{2} \bigg)\ln\bigg(  k+\frac{1+\nu_{\min} }{2} \bigg) \Bigg\}.
\end{align*}
Simplification gives (see \eqref{asymptotics for X2 a}, \eqref{asymptotics for X2 b}, and \eqref{asymptotics for X2 d})
\begin{align} 
\int_{\nu_{\min}}^{\nu_{\ell}} X_{2,\nu_{\ell}'}^{(c)} d\nu_{\ell}' = &\; \frac{\ln(2\pi)}{2}(\nu_{\min}-\nu_{\ell}) + \ln G(1+\nu_{\ell})-\ln G(1+\nu_{\min}) \nonumber
\\ 
& + \frac{1+2\nu_\ell-\nu_{\min} }{2} \ln\bigg(  \frac{1+2\nu_\ell-\nu_{\min} }{2} \bigg) -  \frac{1+\nu_{\min} }{2} \ln\bigg(  \frac{1+\nu_{\min} }{2} \bigg) \nonumber
\\
& -\zeta'\bigg( -1;\frac{1+2\nu_{\ell}-\nu_{\min}}{2} \bigg)+\zeta'\bigg( -1;\frac{1+\nu_{\min}}{2}\bigg)  +\frac{\nu_{\ell}}{2} (\nu_{\ell}-\nu_{\min}).\label{X2nu after integrating}
\end{align}
On the other hand, part (b) of Lemma \ref{lemma: differential identities} and Propositions \ref{prop: I_1alpha}, \ref{prop: asymptotics Xalpha}, \ref{prop: asymptotics of Zalpha}, and \ref{prop: asymptotics I34alpha} imply that
\begin{align} \nonumber
\partial_{\nu_{\ell}}(\ln C) -X_{2,\nu_{\ell}}^{(c)}=&\; \ln \Gamma \bigg( \frac{1+2 \nu_{\ell}-\nu_{\min}}{2} \bigg)-\frac{\ln(2\pi)}{2}
	\\ \nonumber
& + \bigg(\sum_{\substack{j=1 \\ j \neq \ell}}^r\nu_j -\sum_{k=1}^{q}\mu_k  \bigg)  \ln(r-q)+\nu_{\ell}\frac{1+r-q-(r-q)^2 }{1+r-q} \ln(r-q)
	\\
& - \bigg(\sum_{\substack{j=1 \\ j \neq \ell}}^r \nu_j -\sum_{k=1}^{q}\mu_k  \bigg) \ln(1+r-q) -(2+q-r)\nu_{\ell} \ln(1+r-q).\label{cosntantnu before integrating}
\end{align}
Using the identity \eqref{identity int ln Gamma}, we  can integrate the term in (\ref{cosntantnu before integrating}) involving the Gamma function:
\begin{align*}
\int_{\nu_{\min}}^{\nu_{\ell}}  \ln \Gamma \bigg( \frac{1+2 \nu_{\ell}'-\nu_{\min}}{2} \bigg)d\nu_{\ell}' &= \frac{\ln(2\pi)}{2}(\nu_{\ell}-\nu_{\min} )  - \frac{\nu_{\ell}}{2}(\nu_{\ell}-\nu_{\min} ) 
\\
&\quad - \frac{1-2 \nu_{\ell}+\nu_{\min}}{2} \ln \Gamma \bigg(  \frac{1+2 \nu_{\ell}-\nu_{\min}}{2} \bigg) + \frac{1-\nu_{\min}}{2} \ln \Gamma \bigg(   \frac{1+\nu_{\min}}{2}\bigg)
\\
&\quad -\ln G \bigg(   \frac{1+2 \nu_{\ell}-\nu_{\min}}{2}\bigg) + \ln G \bigg(  \frac{1+\nu_{\min}}{2} \bigg).
\end{align*}
Integrating also the other terms in \eqref{cosntantnu before integrating} and utilizing \eqref{X2nu after integrating}, we arrive at
\begin{align*}
& \int_{\nu_{\min}}^{\nu_{\ell}}\partial_{\nu_{\ell}'}(\ln C) d\nu_{\ell}' = \frac{\ln(2\pi)}{2}(\nu_{\ell}-\nu_{\min} )  - \frac{\nu_{\ell}}{2}(\nu_{\ell}-\nu_{\min} ) - \frac{1-2 \nu_{\ell}+\nu_{\min}}{2} \ln \Gamma \bigg(  \frac{1+2 \nu_{\ell}-\nu_{\min}}{2} \bigg)
\\
&\quad  + \frac{1-\nu_{\min}}{2} \ln \Gamma \bigg(   \frac{1+\nu_{\min}}{2}\bigg) -\ln G \bigg(   \frac{1+2 \nu_{\ell}-\nu_{\min}}{2}\bigg) + \ln G \bigg(  \frac{1+\nu_{\min}}{2} \bigg)  - \frac{\ln(2\pi)}{2}(\nu_{\ell}-\nu_{\min}) 
\\
&\quad+ (\nu_{\ell}-\nu_{\min})\bigg(\sum_{\substack{j=1 \\ j\neq \ell}}^r \nu_j -\sum_{k=1}^{q}\mu_k  \bigg)\big(  \ln(r-q) - \ln(1+r-q)\big) +\frac{\ln(2\pi)}{2}(\nu_{\min}-\nu_{\ell} )
\\
&\quad +\frac{\nu_{\ell}^2-\nu_{\min}^2}{2}\frac{1+r-q-(r-q)^2 }{1+r-q} \ln(r-q) -(2+q-r)\frac{\nu_{\ell}^2-\nu_{\min}^2}{2} \ln(1+r-q)
\\
&\quad  +\ln G(1+\nu_{\ell})-\ln G(1+\nu_{\min}) + \frac{1+2\nu_\ell-\nu_{\min} }{2} \ln\bigg(  \frac{1+2\nu_\ell-\nu_{\min} }{2} \bigg) -  \frac{1+\nu_{\min} }{2} \ln\bigg(  \frac{1+\nu_{\min} }{2} \bigg) \nonumber
\\
&\quad -\zeta'\bigg( -1;\frac{1+2\nu_{\ell}-\nu_{\min}}{2} \bigg)+\zeta'\bigg( -1;\frac{1+\nu_{\min}}{2}\bigg) +\frac{\nu_{\ell}}{2} (\nu_{\ell}-\nu_{\min}).
\end{align*}
In view of \eqref{identity Hurwitz Barnes Gamma} and the identities $\Gamma(z+1) = z\Gamma(z)$ and $G(z+1)= \Gamma(z)G(z)$, this expression simplifies to \eqref{after integrating in nu}. The proof of \eqref{after integrating in mu} is analogous.
\end{proof}

\section{Proof of Theorem \ref{mainthm}} \label{Section: proof main thm}
We use the strategy described in Section \ref{outlinesubsec} to prove Theorem \ref{mainthm}. Let $C^{(\ell)}$ be the multiplicative constant arising in the large gap asymptotics for the point process induced by $\mathbb{K}^{(\ell)}$, $\ell \in \{0,\ldots,r+q\}$, where $\mathbb{K}^{(\ell)}$ is given by \eqref{def of K ell}.
The final constant is given by
\begin{align}
\ln C &=\ln C_{r-q} + \sum_{\ell=1}^{r} \int_{\nu_{\min}}^{\nu_{\ell}}\partial_{\nu_{\ell}'}(\ln C^{(\ell)}) d\nu_{\ell}' + \sum_{\ell=1}^{q} \int_{\nu_{\min}}^{\mu_{k}}\partial_{\mu_{\ell}'}(\ln C^{(r+\ell)}) d\mu_{\ell}', \label{sum for final constant}
\end{align}
where $C_{r-q}$ is the constant in (\ref{asymptotics det Kr}) associated with $\mathbb{K}_{r-q}$. 
Proposition \ref{prop: asymptotics of Kr} with $r$ replaced by $r-q$ gives
\begin{align} \nonumber
& \ln C_{r-q}=(r-q)\ln G(1+\nu_{\min}) - \frac{\ln(2\pi)}{2}\nu_{\min}(r-q)- (r-q-1) \zeta'(-1)
\\ \label{value Cr-q}
& +\frac{-2+(r-q)^2(r-q-1+12\nu_{\min}^2)}{24(r-q+1)} \ln(r-q) - \frac{(r-q-1)^2+12(r-q)\nu_{\min}^2}{24}\ln(1+r-q).
\end{align}
On the other hand, Proposition \ref{prop: constant after deforming numu} shows that, for $\ell = 1,\ldots,r$,
\begin{align} \nonumber
\int_{\nu_{\min}}^{\nu_{\ell}} & \partial_{\nu_{\ell}'}(\ln C^{(\ell)}) d\nu_{\ell}' =(\nu_{\ell}-\nu_{\min})\bigg(\sum_{j=1}^{\ell-1} \nu_j + (r-q-\ell) \nu_{\min} \bigg)\big(  \ln(r-q) - \ln(1+r-q)\big)
\\ \nonumber
&\quad +\frac{\nu_{\ell}^2-\nu_{\min}^2}{2}\frac{1+r-q-(r-q)^2 }{1+r-q} \ln(r-q) -(2+q-r)\frac{\nu_{\ell}^2-\nu_{\min}^2}{2} \ln(1+r-q)
\\ \label{integrating step nuell}
&\quad +\ln G(1+\nu_{\ell})-\ln G(1+\nu_{\min}) +(\nu_{\min}-\nu_{\ell} ) \frac{\ln(2\pi)}{2}
\end{align}
and, for $\ell = 1,\ldots,q$,
\begin{align} \nonumber
 \int_{\nu_{\min}}^{\mu_{\ell}}\partial_{\mu_{\ell}'}(\ln C^{(r+\ell)}) d\mu_{\ell}' &=(\nu_{\min}-\mu_\ell)\bigg(\sum_{j=1}^r \nu_j -\sum_{k=1}^{\ell-1}\mu_k-(q-\ell) \nu_{\min}  \bigg)\big(  \ln(r-q) - \ln(1+r-q)\big)
 \\ \nonumber
 &\quad+\frac{\mu_\ell^2 -\nu_{\min}^2}{2}\frac{1+r-q+(r-q)^2 }{1+r-q} \ln(r-q) -(r-q)\frac{\mu_\ell^2 -\nu_{\min}^2}{2} \ln(1+r-q)
 \\ \label{integrating step muell}
 &\quad -\ln G(1+\mu_{\ell})+\ln G(1+\nu_{\min}) -(\nu_{\min}-\mu_\ell)\frac{\ln(2\pi)}{2}.
\end{align}
By substituting \eqref{value Cr-q}, \eqref{integrating step nuell}, and \eqref{integrating step muell} into \eqref{sum for final constant}, we find the expression \eqref{constant C2}. This completes the proof of Theorem \ref{mainthm}. \hfill{$\square$}

\end{document}